\documentclass[10pt,journal,compsoc]{IEEEtran}

\usepackage{amsmath}%, amsthm, amssymb}

\usepackage{xspace}

\usepackage{subfigure}                              % LaTeX 2e
\usepackage{psfrag}
\usepackage{epsfig}
\usepackage{tabularx}
\usepackage{graphicx}
\usepackage{multirow}
\usepackage{algorithm,algpseudocode}
\usepackage[english]{babel}
\usepackage{array}
\usepackage{psfrag}
\usepackage{pstricks,pst-node}
\usepackage{pstricks-add}
\usepackage{bm}
\usepackage{amssymb}
\usepackage{authblk}
\usepackage{ragged2e}

\newtheorem{definition}{\textbf{Definition}}
\newtheorem{example}{\textbf{Example}}
\newtheorem{theorem}{\textbf{Theorem}}

\newtheorem{lemma}{\textbf{Lemma}}

\newtheorem{proposition}{\textbf{Proposition}}
\newtheorem{fact}{\textbf{Fact}}
\newtheorem{rule0}{\textbf{Rule}}

 \newenvironment{proof}
         { \noindent {\bf Proof.} }
         { \begin{flushright}$\Box$\end{flushright} }
\newcommand{\secref}[1]{Section~\ref{#1}}

\newcommand{\figref}[1]{Fig.~\ref{#1}}
\newcommand{\tabref}[1]{Tab.~\ref{#1}}
\newcommand{\eqnref}[1]{Eqn.~(\ref{#1})}
\newcommand{\eqnsref}[1]{Eqns.~(\ref{#1})}

\newcommand{\defref}[1]{Def.~\ref{#1}}
\newcommand{\thmref}[1]{Thm.~\ref{#1}}
\newcommand{\propref}[1]{Prop.~\ref{#1}}
\newcommand{\lemref}[1]{Lem.~\ref{#1}}

\newcommand{\proporef}[1]{Propo.~\ref{#1}}
\newcommand{\ruleref}[1]{Rule~\ref{#1}}
\newcommand{\factref}[1]{Fact~\ref{#1}}
\newcommand{\floor}[1]{\lfloor{#1}\rfloor}

\newcommand{\V}{\ensuremath{\mathcal{V}}\xspace}

\newcommand{\GA}{\gamma}
\newcommand{\GAH}{\hat{\gamma}}
\newcommand{\KA}{\kappa}
\newcommand{\PII}{\pi}
\newcommand{\ALP}{\alpha}
\newcommand{\ALPB}{\bar{\alpha}}
\newcommand{\ALPBH}{\hat{\bar{\alpha}}}
\newcommand{\BE}{\beta}
\newcommand{\BEB}{\bar{\beta}}
\newcommand{\SIG}{\sigma}
\newcommand{\SIB}{\bar{\sigma}}

\newcommand{\inv}{^{-1}}
\newcommand{\VCC}[1][noarg]{\ifthenelse{\equal{#1}{noarg}}{\ensuremath{\V}}{\ensuremath{\V_{#1}}}\xspace}
\newcommand{\VCCi}[1][noarg]{\ifthenelse{\equal{#1}{noarg}}{\ensuremath{\V\inv}}{\ensuremath{\V_{#1}\inv}}\xspace}
\newcommand{\uVCC}[1][noarg]{\ifthenelse{\equal{#1}{noarg}}{\ensuremath{\V^u}}{\ensuremath{\V_{#1}^u}}\xspace}
\newcommand{\lVCC}[1][noarg]{\ifthenelse{\equal{#1}{noarg}}{\ensuremath{\V^l}}{\ensuremath{\V_{#1}^l}}\xspace}
\newcommand{\uVCCi}[1][noarg]{\ifthenelse{\equal{#1}{noarg}}{\ensuremath{\V^{u\inv}}}{\ensuremath{\V_{#1}^{u\inv}}}\xspace}
\newcommand{\lVCCi}[1][noarg]{\ifthenelse{\equal{#1}{noarg}}{\ensuremath{\V^{l\inv}}}{\ensuremath{\V_{#1}^{l\inv}}}\xspace}

\newcommand{\G}[1][noarg]{\ifthenelse{\equal{#1}{noarg}}{\ensuremath{\GA}}{\ensuremath{\GA_{#1}}}\xspace}
\newcommand{\uG}[1][noarg]{\ifthenelse{\equal{#1}{noarg}}{\ensuremath{\GA^u}}{\ensuremath{\GA_{#1}^u}}\xspace}
\newcommand{\lG}[1][noarg]{\ifthenelse{\equal{#1}{noarg}}{\ensuremath{\GA^l}}{\ensuremath{\GA_{#1}^l}}\xspace}
\newcommand{\Gi}[1][noarg]{\ifthenelse{\equal{#1}{noarg}}{\ensuremath{\GA\inv}}{\ensuremath{\GA_{#1}\inv}}\xspace}
\newcommand{\uGi}[1][noarg]{\ifthenelse{\equal{#1}{noarg}}{\ensuremath{\GA^{u\inv}}}{\ensuremath{\GA_{#1}^{u\inv}}}\xspace}
\newcommand{\lGi}[1][noarg]{\ifthenelse{\equal{#1}{noarg}}{\ensuremath{\GA^{l\inv}}}{\ensuremath{\GA_{#1}^{l\inv}}}\xspace}
\newcommand{\uGh}[1][noarg]{\ifthenelse{\equal{#1}{noarg}}{\ensuremath{\GAH^u}}{\ensuremath{\GAH_{#1}^u}}\xspace}

\newcommand{\K}[1][noarg]{\ifthenelse{\equal{#1}{noarg}}{\ensuremath{\KA}}{\ensuremath{\KA_{#1}}}\xspace}
\newcommand{\uK}[1][noarg]{\ifthenelse{\equal{#1}{noarg}}{\ensuremath{\KA^u}}{\ensuremath{\KA_{#1}^u}}\xspace}
\newcommand{\lK}[1][noarg]{\ifthenelse{\equal{#1}{noarg}}{\ensuremath{\KA^l}}{\ensuremath{\KA_{#1}^l}}\xspace}
\newcommand{\uKi}[1][noarg]{\ifthenelse{\equal{#1}{noarg}}{\ensuremath{\KA^{u\inv}}}{\ensuremath{\KA_{#1}^{u\inv}}}\xspace}
\newcommand{\lKi}[1][noarg]{\ifthenelse{\equal{#1}{noarg}}{\ensuremath{\KA^{l\inv}}}{\ensuremath{\KA_{#1}^{l\inv}}}\xspace}

\newcommand{\PI}[1][noarg]{\ifthenelse{\equal{#1}{noarg}}{\ensuremath{\PII}}{\ensuremath{\PII_{#1}}}\xspace}
\newcommand{\uP}[1][noarg]{\ifthenelse{\equal{#1}{noarg}}{\ensuremath{\PII^u}}{\ensuremath{\PII_{#1}^u}}\xspace}
\newcommand{\lP}[1][noarg]{\ifthenelse{\equal{#1}{noarg}}{\ensuremath{\PII^l}}{\ensuremath{\PII_{#1}^l}}\xspace}
\newcommand{\uPi}[1][noarg]{\ifthenelse{\equal{#1}{noarg}}{\ensuremath{\PII^{u\inv}}}{\ensuremath{\PII_{#1}^{u\inv}}}\xspace}
\newcommand{\lPi}[1][noarg]{\ifthenelse{\equal{#1}{noarg}}{\ensuremath{\PII^{l\inv}}}{\ensuremath{\PII_{#1}^{l\inv}}}\xspace}

\newcommand{\AL}[1][noarg]{\ifthenelse{\equal{#1}{noarg}}{\ensuremath{\ALP}}{\ensuremath{\ALP_{#1}}}\xspace}
\newcommand{\Ab}[1][noarg]{\ifthenelse{\equal{#1}{noarg}}{\ensuremath{\ALPB}}{\ensuremath{\ALPB_{#1}}}\xspace}
\newcommand{\uA}[1][noarg]{\ifthenelse{\equal{#1}{noarg}}{\ensuremath{\ALP^u}}{\ensuremath{\ALP_{#1}^u}}\xspace}
\newcommand{\lA}[1][noarg]{\ifthenelse{\equal{#1}{noarg}}{\ensuremath{\ALP^l}}{\ensuremath{\ALP_{#1}^l}}\xspace}
\newcommand{\uAi}[1][noarg]{\ifthenelse{\equal{#1}{noarg}}{\ensuremath{\ALP^{u\inv}}}{\ensuremath{\ALP_{#1}^{u\inv}}}\xspace}
\newcommand{\lAi}[1][noarg]{\ifthenelse{\equal{#1}{noarg}}{\ensuremath{\ALP^{l\inv}}}{\ensuremath{\ALP_{#1}^{l\inv}}}\xspace}
\newcommand{\uAb}[1][noarg]{\ifthenelse{\equal{#1}{noarg}}{\ensuremath{\ALPB^u}}{\ensuremath{\ALPB_{#1}^u}}\xspace}
\newcommand{\lAb}[1][noarg]{\ifthenelse{\equal{#1}{noarg}}{\ensuremath{\ALPB^l}}{\ensuremath{\ALPB_{#1}^l}}\xspace}
\newcommand{\uAbi}[1][noarg]{\ifthenelse{\equal{#1}{noarg}}{\ensuremath{\ALPB^{u\inv}}}{\ensuremath{\ALPB_{#1}^{u\inv}}}\xspace}
\newcommand{\lAbi}[1][noarg]{\ifthenelse{\equal{#1}{noarg}}{\ensuremath{\ALPB^{l\inv}}}{\ensuremath{\ALPB_{#1}^{l\inv}}}\xspace}
\newcommand{\uAbh}[1][noarg]{\ifthenelse{\equal{#1}{noarg}}{\ensuremath{\ALPBH^u}}{\ensuremath{\ALPBH_{#1}^u}}\xspace}

\newcommand{\B}[1][noarg]{\ifthenelse{\equal{#1}{noarg}}{\ensuremath{\BE}}{\ensuremath{\BE_{#1}}}\xspace}
\newcommand{\Bb}[1][noarg]{\ifthenelse{\equal{#1}{noarg}}{\ensuremath{\BEB}}{\ensuremath{\BEB_{#1}}}\xspace}
\newcommand{\uB}[1][noarg]{\ifthenelse{\equal{#1}{noarg}}{\ensuremath{\BE^u}}{\ensuremath{\BE_{#1}^u}}\xspace}
\newcommand{\lB}[1][noarg]{\ifthenelse{\equal{#1}{noarg}}{\ensuremath{\BE^l}}{\ensuremath{\BE_{#1}^l}}\xspace}
\newcommand{\uBi}[1][noarg]{\ifthenelse{\equal{#1}{noarg}}{\ensuremath{\BE^{u\inv}}}{\ensuremath{\BE_{#1}^{u\inv}}}\xspace}
\newcommand{\lBi}[1][noarg]{\ifthenelse{\equal{#1}{noarg}}{\ensuremath{\BE^{l\inv}}}{\ensuremath{\BE_{#1}^{l\inv}}}\xspace}
\newcommand{\uBb}[1][noarg]{\ifthenelse{\equal{#1}{noarg}}{\ensuremath{\BEB^u}}{\ensuremath{\BEB_{#1}^u}}\xspace}
\newcommand{\lBb}[1][noarg]{\ifthenelse{\equal{#1}{noarg}}{\ensuremath{\BEB^l}}{\ensuremath{\BEB_{#1}^l}}\xspace}
\newcommand{\uBbi}[1][noarg]{\ifthenelse{\equal{#1}{noarg}}{\ensuremath{\BEB^{u\inv}}}{\ensuremath{\BEB_{#1}^{u\inv}}}\xspace}
\newcommand{\lBbi}[1][noarg]{\ifthenelse{\equal{#1}{noarg}}{\ensuremath{\BEB^{l\inv}}}{\ensuremath{\BEB_{#1}^{l\inv}}}\xspace}

\newcommand{\SI}[1][noarg]{\ifthenelse{\equal{#1}{noarg}}{\ensuremath{\SIG}}{\ensuremath{\SIG_{#1}}}\xspace}
\newcommand{\Sb}[1][noarg]{\ifthenelse{\equal{#1}{noarg}}{\ensuremath{\SIB}}{\ensuremath{\SIB_{#1}}}\xspace}
\newcommand{\uS}[1][noarg]{\ifthenelse{\equal{#1}{noarg}}{\ensuremath{\SIG^u}}{\ensuremath{\SIG_{#1}^u}}\xspace}
\newcommand{\lS}[1][noarg]{\ifthenelse{\equal{#1}{noarg}}{\ensuremath{\SIG^l}}{\ensuremath{\SIG_{#1}^l}}\xspace}
\newcommand{\uSb}[1][noarg]{\ifthenelse{\equal{#1}{noarg}}{\ensuremath{\SIB^u}}{\ensuremath{\SIB_{#1}^u}}\xspace}
\newcommand{\lSb}[1][noarg]{\ifthenelse{\equal{#1}{noarg}}{\ensuremath{\SIB^l}}{\ensuremath{\SIB_{#1}^l}}\xspace}

\newcommand{\task}[1][noarg]{\ifthenelse{\equal{#1}{noarg}}{\ensuremath{\mathrm{T}}}{\ensuremath{\mathrm{T{#1}}}}\xspace}

\newcommand{\PE}[1][noarg]{\ifthenelse{\equal{#1}{noarg}}{\ensuremath{\mathrm{PE}}}{\ensuremath{\mathrm{PE{#1}}}}\xspace}

\newboolean{clean}\setboolean{clean}{false}

\newcommand{\todo}[1]{%
    \ifthenelse{\boolean{clean}}{}{
        \noindent\color{blue}[{\small\textsc{$\bigstar$~ToDo(#1):}\/}]%
        \color{black}
    }
}%

\newcommand{\storyline}[1]{%
    \ifthenelse{\boolean{clean}}{}{
        \noindent\color{blue}{$\bigstar$ #1\/}%
        \color{black}
    }
}

\newcommand{\p}[1]{%
    \ifthenelse{\boolean{clean}}{}{
        \storyline{#1}
    }
}

\newcommand{\ask}[1]{%
    \ifthenelse{\boolean{clean}}{}{
        \color{blue}[{\small\it\textsc{$\bigstar$ (?):} #1 \/}]%
        \color{black}
    }
}%

\newcommand{\trans}[1]{%
    \ifthenelse{\boolean{clean}}{}{
        \color{green}#1%
        \color{black}
    }
}%

\newcommand{\pic}[1]{%
    \ifthenelse{\boolean{clean}}{}{
        \begin{center}%
        \color{red}[{\small\it\textsc{ Picture:} #1\/}]%
        \color{black}
        \end{center}
    }
}%

\newcommand{\ideas}{%
    \ifthenelse{\boolean{clean}}{}{
        \begin{center}
            \vspace{5mm}%
            \color{red}\texttt{\footnotesize --------------------} {\it\small ideas} \texttt{\footnotesize --------------------}
            \vspace{2mm}%
        \end{center}%
    }
}%

\newcommand{\testtext}{%
    \ifthenelse{\boolean{clean}}{}{
    \color{magenta}
    The brown {\bf fox} jumps over the lazy \emph{dog}. The quick
    brown fox jumps over the lazy dog. The quick brown fox jumps over
    the lazy dog. The quick brown fox jumps over the lazy dog. The
    quick brown fox jumps over the lazy dog.
    \color{black}
    }
}

\newcommand{\w}[1]{% "find better word(s)"
    \ifthenelse{\boolean{clean}}{}{
        \underline{#1}%
    }
}%

\DeclareMathSizes{9}{8}{1}{1}
\begin{document}

\title{Utilization-Based Scheduling of Flexible Mixed-Criticality Real-Time Tasks\thanks{This paper has been submitted to IEEE Transaction on Computers (TC) on Sept-09th-2016, and revised for two times on Jan-19th-2017 and Aug-28th-2017. The submission number on TC is TC-2016-09-0607. This paper is still under review by TC with minor revision. The screen-shot of submission history is also attached in appendix D. Email: chengang@cse.neu.edu.cn;csguannan@comp.polyu.edu.hk;yi@it.uu.se}}

\author{Gang~Chen,
        Nan~Guan,
        Di~Liu, 
				Qingqiang~He,
				Kai~Huang,
				Todor~Stefanov,
				Wang~Yi
}

\def \lo {low-criticality }
\def \hi {high-criticality }
\def \FMC {FMC }
\def \prealg {FG}
\def \micr {mission-critical }
\def \noncr {non-critical }
\newcommand{\sups}[1]{\sup\{#1\}}
\newcommand{\supst}[1]{\overline{\sup}\{#1\}}

\IEEEtitleabstractindextext{%
\begin{abstract}
\justifying{Mixed-criticality models are an emerging paradigm for the design of real-time systems because of their significantly improved resource efficiency. However, formal mixed-criticality models have traditionally been characterized by two impractical assumptions: once \textit{any} high-criticality task overruns, \textit{all} low-criticality tasks are suspended and \textit{all other} high-criticality tasks are assumed to exhibit high-criticality behaviors at the same time. In this paper, we propose a more realistic mixed-criticality model, called the flexible mixed-criticality (FMC) model, in which these two issues are addressed in a combined manner. In this new model, only the overrun task itself is assumed to exhibit high-criticality behavior, while other high-criticality tasks remain in the same mode as before. The guaranteed service levels of low-criticality tasks are gracefully degraded with the overruns of high-criticality tasks. We derive a utilization-based technique to analyze the schedulability of this new mixed-criticality model under EDF-VD scheduling. During run time, the proposed test condition serves an important criterion for dynamic service level tuning, by means of which the maximum available execution budget for low-criticality tasks can be directly determined with minimal overhead while guaranteeing mixed-criticality schedulability. Experiments demonstrate the effectiveness of the FMC scheme compared with state-of-the-art techniques.
\vskip -1.2em
}
\end{abstract}

% Note that keywords are not normally used for peerreview papers.
}
\vskip -1em

% make the title area
\maketitle
\vskip -2em
\IEEEdisplaynontitleabstractindextext
\vskip -2em
\IEEEpeerreviewmaketitle
\vskip -4em
%%%%%%%%%%%%%%%%%%%%%%%%%%%%%%%%%%%%%%%%%%%%%%%%%%%%%%%%%%%%%%%%%%%%%%%%%%%%%%

\section{Introduction}
A mixed-criticality (MC) system is a system in which tasks with different criticality levels share a computing platform. In MC systems, different degrees of assurance must be provided for tasks with different criticality levels. To improve resource efficiency, MC systems \cite{Vestal2007} often specify different WCETs for each task at all existing criticality levels, with those at higher criticality levels being more pessimistic. Normally, tasks are scheduled with less pessimistic WCETs for resource efficiency. Only when the less pessimistic WCET is violated, the system switches to the \hi mode and \textit{only} tasks with higher criticality levels are guaranteed to be scheduled with pessimistic WCETs thereafter.

There is a large body of research work on specifying and scheduling mixed-criticality systems (see \cite{burns2015mixed} for a comprehensive review). However, to ensure the safety of \hi tasks, the classic MC model~\cite{Baruah2012,sanjoyACM,Baruah2010,Baruah2011RTSS,Baruah2011} applies conservative restrictions to the mode-switching scheme. In the classic MC model, whenever \textit{any} \hi task overruns, \textit{all} \lo tasks are immediately abandoned and \textit{all other} \hi tasks are assumed to exhibit \hi behaviors. This mode-switching scheme is not realistic in the following two important respects. 
\begin{itemize}
\item First, it is pessimistic to immediately abandon all \lo tasks, because \lo tasks require a certain timing performance as well~\cite{ISO262262,Su2014}. 
\item Second, it is pessimistic to bind the mode switches of all \hi tasks together for the scenarios where the mode switches of \hi tasks are naturally independent~\cite{gu,ren}.
%Second, the mode switches of \hi tasks are naturally independent. It is overly pessimistic to bind the mode switches of all \hi tasks together in the analysis, because a scenario in which all \hi tasks simultaneously require additional resources will not occur in practice.
\end{itemize}
Although there has been some research on solving the first problem, i.e., \textit{statically} reserving a certain degraded level of service for \lo execution~\cite{burns2013towards,Su2014,Su2013,HGST14a}, to our knowledge, little work has been done to date to address the second problem.           

In this paper, we propose a flexible MC model (denoted as \textbf{FMC}) on a uni-processor platform, in which the two aforementioned issues are addressed in a combined manner. In FMC, the mode switches of all \hi tasks are independent. A single \hi task that violates its \lo WCET triggers only itself into \hi mode, rather than triggering \textit{all} \hi tasks. All other \hi tasks remain at their previous criticality levels and thus do not require to book additional resources at mode-switching points. In this manner, significant resources can be saved compared with the classic MC model~\cite{Baruah2010,Baruah2011,Baruah2011RTSS}. On the other hand, these saved resources can be used by \lo tasks to improve their service quality. More importantly, the proposed FMC model adaptively tunes the service level for \lo tasks to compensate for the overrun of \hi tasks, thereby allowing the system workload to be balanced with minimal service degradation for \lo tasks. At each independent mode-switching point, the service level for \lo tasks is dynamically updated based on the overruns of \hi tasks. By doing so, the quality of service (QoS) for \lo tasks can be significantly improved. 

Since the service level for \lo tasks is dynamically determined during run time, the decision-making procedure should be light-weighted. For this purpose, utilization-based scheduling is more desirable for run-time decision-making because of its simplicity. However, using utilization-based scheduling for our FMC model brings new challenges due to the intrinsic dynamics of this model, such as the service level tuning strategy. In particular, utilization-based schedulability analysis relies on whether the cumulative execution time of \lo tasks can be \textit{effectively} upper bounded. In FMC, the service levels for \lo tasks are dynamically tuned at each mode switching point. Therefore, the cumulative execution time of \lo tasks strongly depends on when mode switches occur. In general, such information is difficult to explicitly represent prior to real execution, because the independence of the mode switches in FMC results in a large analysis space. It is computationally infeasible to analyze all possibilities. To resolve this challenge, we propose a novel approach based on mathematical induction, which allows the cumulative execution time of \lo tasks to be \textit{effectively} upper bounded.   

In this work, we study the schedulability of the proposed \FMC model under EDF-VD scheduling. A utilization-based schedulability test condition is derived by integrating the independent triggering scheme and the adaptive service level tuning scheme. A formal proof of the correctness of this new schedulability test condition is presented. Based on this test condition, an EDF-VD-based MC scheduling algorithm, called FMC-EDF-VD, is proposed for the scheduling of an FMC task system. During run time, the optimal service level for \lo tasks can be directly determined via this condition with minimum overhead, and mixed-criticality schedulability can be simultaneously guaranteed. In addition, we explore the feasible region of the virtual deadline factor for FMC model. Simulation results show that FMC-EDF-VD provides benefits in supporting \lo execution compared with state-of-the-art algorithms.  
%The remainder of this article is organized as follows: \secref{sec:rw} reviews related work in the literature. \secref{sec:model} presents our system model and settings. \secref{sec:overview} provides an overview of the scheduling algorithm and describes its motivation. \secref{sec:proof1} proves the correctness of the proposed scheduling algorithm. \secref{sec:slts} discusses two service level tuning strategies. An experimental evaluation is presented in \secref{sec:eva}, and \secref{sec:conclusion} concludes the paper.
\vskip -2em

%%%%%%%%%%%%%%%%%%%%%%%%%%%%%%%%%%%%%%%%%%%%%%%%%%%%%%%%%%%%%%%%%%%%%%%%%%%%%%%%%  
\section{Related Work}
\label{sec:rw}
Mixed-criticality scheduling is a research field that has received considerable attention in recent years. As stated in \cite{burns2013towards}, much existing research work~\cite{Baruah2010,Baruah2011,Baruah2011RTSS} on MC scheduling makes the pessimistic assumption that all \lo tasks are immediately abandoned once the system enters \hi mode. Instead of abandoning all \lo tasks, some efforts~\cite{burns2013towards,Su2014,Su2013,HGST14a,DiLiu} have been made to provide solutions for offering \lo tasks a certain degraded service quality when the system is in \hi mode. Nevertheless, these studies still use a pessimistic mode-switch triggering scheme in which, whenever one \hi task overruns, all other \hi tasks are triggered to exhibit \hi behavior and book unnecessary resources.

Recent work presented in \cite{gu,ren,huang} offers solutions for improving performance for \lo tasks by using different mode-switch triggering strategies. Huang et al. \cite{huang} proposed an interference constraint graph to specify the execution dependencies between \hi and \lo tasks. However, this approach still uses high-confidence WCET estimates for all \hi tasks when determining system schedulability, and therefore does not address the second problem discussed above. Gu et al. \cite{gu} presented a component-based strategy in which the component boundaries offer the isolation necessary to support the execution of \lo tasks. Minor overruns can be handled with an internal mode switch by dropping off all \lo jobs within a component. More extensive overruns will result in a system-wide external mode switch and the dropping off of all \lo jobs. Therefore, the mode switches at the internal and external levels still use pessimistic strategy in which all \lo tasks are abandoned once a mode switch occurs at the corresponding level. The two problems mentioned above still exist at both levels. In addition, the system schedulability is tested using a demand bound function (DBF) based approach. The complexity of the schedulability test is exponential in the size of the input~\cite{gu}, resulting in costly computations.    

Ren and Phan \cite{ren} proposed a partitioned scheduling algorithm based on group-based Pfair-like scheduling~\cite{Pfairs} for mixed-criticality systems. Within a task group, a single \hi task is encapsulated with several \lo tasks. The tasks are scheduled via Pfair-like scheduling~\cite{Pfairs} by breaking them into quantum-length sub-tasks. Sub-tasks that belong to different groups are scheduled on an earliest-pseudo-deadline-first (EPDF) basis. Pfair scheduling is a well-known optimal scheduling method for scheduling periodic real-time tasks on a multiple-resource system. However, Pfair scheduling poses many practical problems~\cite{Pfairs}. First,  the Pfair algorithm incurs very high scheduling overhead because of frequent preemptions caused by the small quantum lengths. Second, the task groups are explicitly required to be well synchronized and to make progress at a steady rate~\cite{ZhuParis}. Therefore, the work presented in \cite{ren} strongly relies on the periodic task models. In addition, the system schedulability in \cite{ren} is determined by solving a MINLP problem, which in general has NP-hard complexity\cite{MINLP}. Because of  this complexity, the scalability problem needs to be carefully considered.

Compared with the existing work \cite{gu,ren}, the proposed FMC model and its scheduling techniques offer both \textbf{simplicity and flexibility}. In particular, our work differs from these approaches in the following respects. Compared with the Pfair-based scheduling method~\cite{ren} which relies on periodic task models, our paper derives an EDF-VD-based scheduling scheme for sporadic mixed-criticality task systems, that incorporates an independent mode-switch triggering scheme and an adaptive service level tuning scheme. EDF-VD has shown strong competence in both theoretical and empirical evaluations~\cite{Baruah2012}. Compared with the work presented in \cite{gu}, our approach uses a more flexible strategy that allows a component/system to abandon  \lo tasks in accordance with run-time demands. Therefore, both of the problems stated above are addressed in our approach. In contrast to the work of \cite{gu,ren}, our approach is based on a utilization-based schedulability analysis. The system schedulability can be effectively determined. From the designer's perspective, our utilization-based approach requires simpler specifications and reasoning compared with the work of \cite{ren,gu}. In terms of flexibility, our approach can efficiently allocate execution budgets for \lo tasks during runtime in accordance with demands, whereas the approaches presented in \cite{gu,ren} require that \lo tasks should be executed in accordance with the dependencies between \lo and \hi tasks that have been determined in off-line. 
\vskip -2em

%%%%%%%%%%%%%%%%%%%%%%%%%%%%%%%%%%%%%%%%%%%%%%%%%%%%%%%%%%%%%%%%%%%%%%%%%%%%%%%%%%%%%%%%%%%%%%
\section{System Models and Background}
\label{sec:model}
%In this section, we introduce the system models and present the preliminary knowledge that will be used throughout this paper.
\subsection{FMC implicit-deadline sporadic task model}
\noindent\textbf{Task model}: We consider an MC system with two different criticality levels, \textit{HI} and \textit{LO}. The task set $\gamma$ contains $n$ MC implicit-deadline sporadic tasks which are scheduled on a uni-processor platform. Each task $\tau_i$ in $\gamma$ generates an infinite sequence of jobs and can be specified by a tuple $\{T_i,L_i,\mathcal{C}_i\}$. Here, $T_i$ denotes the minimum job-arrival intervals. $L_i\in\{LO,HI\}$ denotes the criticality level of a task. Each task is either a \lo task or \hi task. $\gamma_{LO}$ and $\gamma_{HI}$ (where $\gamma = \gamma_{LO} \cup \gamma_{HI}$) denote \lo task set and \hi task set, respectively. $\mathcal{C}_i \in \{C_i^{LO},C_i^{HI}\}$ is the list of WCETs, where $C_i^{LO}$ and $C_i^{HI}$ denote the \lo and \hi WCETs, respectively. 

For \hi tasks, the WCETs satisfy $C_i^{LO}<C_i^{HI}$. For \lo tasks, their execution budget is dynamically determined in FMC based on the overruns of \hi tasks. To characterize the execution behavior of \lo tasks in \hi mode, we now introduce the concept of the service level on each mode-switching point, which specifies the guaranteed service quality after the mode switch. 

\noindent\textbf{Service level}: Instead of completely discarding all \lo tasks, Burns and Baruah in \cite{burns2013towards} proposed a more practical MC task model in which \lo tasks are allowed to statically reserve resources for their execution
at a degraded service level in \hi mode (i.e., a reduced execution budget). By contrast, in FMC, the execution budget is dynamically determined based on the run-time overruns rather than statically reserved as in \cite{burns2013towards}. To model this dynamic behavior, the service level concept defined in \cite{burns2013towards} should be extended to apply to  independent mode switches. Therefore, we define the service level $z_i^k$ when the system has undergone \textit{$k$} mode switches. 
\begin{definition}
\label{def:6}
(Service level $z_i^k$ when \textit{$k$} mode switches have occurred).
If \lo task $\tau_i$ is executed at service level $z_i^k$ when the system has undergone \textit{$k$} mode switches, up to $z_i^k\cdot{}C_i^{LO}$ time units can be used for the execution of $\tau_i$ in one period $T_i$. When $\tau_i$ runs in \lo mode, we say $\tau_i$ is executed at service level $z_i^0$, where $z_i^0=1$.   
\end{definition} 
%The service level definition given above is compliant with the concept of the \textit{imprecise computation model} developed by Lin et al.\cite{LinIMC} to deal with time-constrained iterative calculations. The \textit{imprecise computation model} has many potential real-life applications, such as optimal control~\cite{RateIMC} and fault-tolerant scheduling~\cite{Han2003A} problems.
The service level definition given above is compliant with the concept of the \textit{imprecise computation model} developed by Lin et al.\cite{LinIMC} to deal with time-constrained iterative calculations. 
\textit{Imprecise computation model} is partly motivated by the observation that many real-time computations are iterative in nature, solving a numeric problem by successive approximations. Terminating an iteration early can return useful imprecise results. With this motivation in mind, the imprecise computation model can be used in a natural way to enhance graceful degradation~\cite{LinIMC1}.
The practicality of \textit{imprecise computation model} has been deeply investigated and verified in \cite{IMCR3}. Imprecise computation model provides an approximate but timely result, which may be acceptable in many application areas. Examples of such applications are optimal control~\cite{RateIMC}, multimedia applications~\cite{Rajkumar1997A}, image and speech processing~\cite{Feng}, and fault-tolerant scheduling problems~\cite{Han2003A}.
In FMC, when an overrun occurs, \lo tasks will be terminated before completion and sacrifice the quality of the produced results to ensure their timing correctness.
%The \lo executions are discarded in a controlled way for performance optimization. 
%In FMC, each \lo task $\tau_i$ has various versions of a given function~\cite{mversion,LinIMC}. These versions all perform the same task, but consume different amounts of execution time and thus produce different output qualities~\cite{mversion}. If an overrun occurs in the system, a \lo task can sacrifice the quality of the produced results to ensure their timing correctness, by executing an alternative version appropriate to the available service level.

\noindent\textbf{Assumptions}: For the remainder of the manuscript, we make the following assumptions: (1) Regarding the compensation for the $k^{\textit{th}}$ overrun of a \hi task, we assume that $z_i^k\le{z_i^{k-1}}$. After the $k^{\textit{th}}$ mode-switching point, the allowed execution time budget in one period should thus be reduced from $z_i^{k-1}\cdot{c_i^{LO}}$ to $z_i^{k}\cdot{c_i^{LO}}$. 
(2) According to~\cite{Baruah2012}, if $u_{LO}^{LO}+u_{HI}^{HI}\le{1}$, then all tasks can be perfectly scheduled by regular EDF under the worst-case reservation strategy. Therefore, we here consider meaningful cases in which $u_{LO}^{LO}+u_{HI}^{HI}>{1}$. 

\noindent\textbf{Utilization}: Low and high utilization for a task $\tau_i$ are defined as $u_i^{LO}=\frac{c_i^{LO}}{T_i}$ and $u_i^{HI}=\frac{c_i^{HI}}{T_i}$, respectively. The system-level utilization for task set $\gamma$ are defined as $u_{LO}^{LO}=\sum_{\tau_i\in{\gamma_{LO}}}u_i^{LO}$, $u_{HI}^{LO}=\sum_{\tau_i\in{\gamma_{HI}}}u_i^{LO}$, and $u_{HI}^{HI}=\sum_{\tau_i\in{\gamma_{HI}}}u_i^{HI}$. The system utilization of \lo tasks after $k^{\textit{th}}$ mode-switching point can be defined as $u_{LO}^{k}=\sum_{\tau_i\in{\gamma_{LO}}}z_i^k\cdot{}u_i^{LO}$. To guarantee the execution of the mandatory portions of \lo tasks, the mandatory utilization can be defined as $u_{LO}^{man}=\sum_{\tau_i\in{\gamma_{LO}}}z_{i}^{man}\cdot{}u_i^{LO}$, where $z_{i}^{man}$ is the mandatory service level for task $\tau_i$ as specified by the users.

\subsection{Execution semantics of the FMC model}
\label{sec:sem}
The main \textbf{differences} between our FMC execution model and the classic MC execution model lie in the independent mode-switch triggering scheme for \hi tasks and the dynamic service tuning of \lo tasks. In contrast to the classic MC model, the FMC model allows an independent triggering scheme in which the overrun of one \hi task triggers only itself into \hi mode. Consequently, the \hi mode of the system in FMC depends on the number of \hi tasks that have overrun. Therefore, we introduce the following definition:
\begin{definition}
\label{def:61}
($k$-level \hi mode). At a given instant of time, if $k$ \hi tasks have entered \hi mode, then the system is in \textbf{$k$-level \hi mode}. For \lo mode, we say that the system is in \textbf{$0$-level \hi mode}.  
\end{definition}    
Based on \defref{def:61}, the execution semantics of the FMC model is illustrated in \figref{fig:sem}. Initially, the system is in \lo mode (i.e., \textbf{$0$-level \hi mode}). Then, the overruns of \hi tasks trigger the system to proceed through the \hi modes one by one until the condition for transitioning back is satisfied. According to \figref{fig:sem}, the execution semantics can be summarized as follows:
\begin{figure}
\centering
\label{fig:sec1} 
\includegraphics[width=0.95\columnwidth,height=0.37\columnwidth]{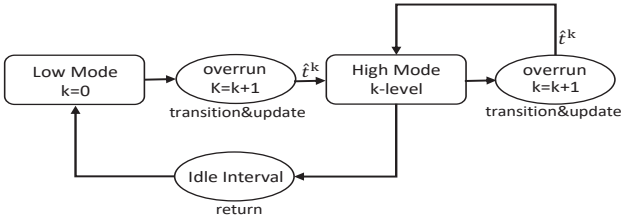}
\vskip -0.5em
 \caption{Execution semantics of the FMC model.}
  \label{fig:sem}
	\vskip -1em
\end{figure}

\begin{itemize}
\item \textbf{Low-criticality mode}: All tasks in $\gamma$ start in $0$-level \hi mode (i.e., \lo mode). As long as no \hi task violates its $C_i^{LO}$, the system remains in $0$-level \hi mode. In this mode, all tasks are scheduled with $C_i^{LO}$. 
\item \textbf{Transition}: When one job of a \hi task \textit{that is being executed in \lo mode} overruns its $C_i^{LO}$, this \hi task immediately switches into \hi mode. However, the overrun of this task does not trigger other \hi tasks to enter \hi mode. All other \hi tasks still remain in the \textbf{same} mode as before. Correspondingly, the system transitions to a higher-level \hi mode\footnote{Without loss of generality, we assume that the system is in $k$-level \hi mode.}.
\item \textbf{Updates}: At the $k^{\textit{th}}$ transition point (corresponding to time instant $\hat{t}^{k}$ in \figref{fig:sem}), a new service level $z_i^{k}$ is determined and updated to provide degraded service for \lo tasks $\tau_i$ to balance the resource demand caused by the overrun of the \hi task. At this time, if any \lo jobs have completed more than $z_i^{k}\cdot{c_i^{LO}}$ time units of execution (i.e., have used up the updated execution budget for the current period), those jobs will be suspended immediately and wait for the budget to be renewed in the next period. Otherwise, \lo jobs can continue to use the remaining time budget for their execution.
\item \textbf{Return to \lo mode}: When the system detects an idle interval~\cite{burns2013towards,Back}, the system will transition back into \lo mode.  
\end{itemize}

\subsection{EDF-VD scheduling}
EDF-VD~\cite{Baruah2012,sanjoyACM} is a scheduling algorithm for implementing classic preemptive EDF scheduling in MC systems. The main concept of EDF-VD is to artificially reduce the (virtual) deadlines of \hi tasks when the system is in \lo mode. These virtual deadlines can be used to cause \hi tasks to finish earlier to ensure that the system can reserve a sufficient execution budget for the \hi tasks to meet their actual deadlines after the system switches into \hi mode. In this paper, we study the schedulability under EDF-VD scheduling for the proposed \FMC model.

%%%%%%%%%%%%%%%%%%%%%%%%%%%%%%%%%%%%%%%%%%%%%%%%%%%%%%%%%%%%%%%%%%%%%%%%%%%%%%%%%%%%%%%%%%%
\section{FMC-EDF-VD scheduling algorithm}
\label{sec:overview}
In this section, we provide an overview of the proposed EDF-VD-based scheduling algorithm for our \FMC model, called FMC-EDF-VD. The proposed scheduling algorithm consists of an \textit{off-line step} and a \textit{run-time step}. We implement the \textit{off-line step} prior to \textit{run time} to select a feasible virtual deadline factor $x$ for tightening the deadlines of \hi tasks. During run time, the service levels $z_i^k$ for \lo tasks are dynamically tuned based on the overrun of \hi tasks. Here, we present the operation flow of FMC-EDF-VD.
    
\noindent\textbf{Off-line step}: In accordance with \thmref{theorem:lowcrit}, we first determine $x$ as $\frac{u_{HI}^{LO}}{1-u_{LO}^{LO}}$. Then, to guarantee the schedulability of FMC-EDF-VD, the determined $x$ value should be validated by testing condition \eqnref{eq:thm5:m0} in \thmref{thm:5}. Note that if condition \eqnref{eq:thm5:m0} is not satisfied, then it is reported that the specified task set cannot be scheduled using FMC-EDF-VD.  
         
\noindent\textbf{Run-time step}: The run-time behavior follows the execution semantics presented in \secref{sec:sem}. In \lo mode, all \hi tasks are scheduled with their virtual deadlines. At each mode-switching point, the following two procedures are triggered:  
\begin{itemize}
\item Reset the deadline of overrun \hi task from its virtual deadline to the actual deadline. The deadline settings of other \hi tasks remain the same as before. 
\item Update the service levels for \lo tasks in accordance with \thmref{thm:3}.       
\end{itemize}

Note that various run-time tuning strategies can be specified by the user as long as the condition in \thmref{thm:3} is satisfied. For the purpose of demonstration, a uniform tuning strategy and a dropping-off strategy are discussed in this paper. Complete descriptions of these strategies are provided in \secref{sec:slts}.
\vskip -1em
%\begin{table}[h]
%\centering
%\caption{Example task set}
%\label{ex:e1}
%\begin{tabular}{|c|c|c|c|c|}
%\hline
         %& $L_i$ & $T_i$  & $C_i^{LO}$  & $C_i^{HI}$   \\ \hline
 %$\tau_1$&  LO   &  $8 $    & $4$           & $$          \\   \hline
 %$\tau_2$&  HI   &  $12$    & $2$           & $3.3$                       \\   \hline
 %$\tau_3$&  HI   &  $6$     & $1$           & $2.1$              \\  \hline
%\end{tabular}
%\end{table}
\begin{table}[h]
\centering
\caption{Example task set}
\vskip -0.5em
\label{ex:e3}
\begin{tabular}{|c|c|c|c|c|}
\hline
                                 & $L_i$ & $T_i$    & $C_i^{LO} $    & $C_i^{HI}$   \\ \hline
 $\tau_1, \tau_2, \tau_3, \tau_4$&  HI   &  $40$     & $3$        & $8$              \\  \hline
 $\tau_5$                        &  LO   &  $200$    & $30$      &           \\   \hline
 $\tau_6$                        &  LO   &  $300$    & $75$      &           \\   \hline
\end{tabular}
\vskip -1.5em
\end{table}
%\begin{table}[h]
%\centering
%\caption{MC task specifications}
%\vskip -0.5em
%\label{ex:e3}
%\begin{tabular}{|c|c|c|c|c|}
%\hline
                                 %& $L_i$ & $T_i$    & $C_i^{LO} $    & $C_i^{HI}$   \\ \hline
 %$\tau_1, \tau_2, \tau_3, \tau_4$&  HI   &  $40$     & $3$        & $8$              \\  \hline
 %$\tau_5,\tau_6, \tau_7, \tau_8$ &  LO   &  $160$    & $16$       &           \\   \hline
%\end{tabular}
%\vskip -1.5em
%\end{table}

\subsection{Motivational example}
In this section, we present a motivation example to show how the global triggering scheme in FMC-EDF-VD can efficiently support \lo task execution. The uniform tuning strategy (see \thmref{thm:aa}), in which all \lo tasks share the same service level setting $z^k$ during run time (i.e., $\forall\tau_i\in{}\gamma_{LO}$, $z_i^k=z^k$), is adopted for this demonstration.   
\begin{example}
\label{example1}
For clarity of presentation, we consider a task set that contains four identical \hi tasks and two \lo tasks, as listed in \tabref{ex:e3}. We specify $u_{LO}^{man}=0$ for demonstration. From \tabref{ex:e3}, one can derive $u_{LO}^{LO}=\frac{2}{5}$, $u_{HI}^{LO}=\frac{3}{10}$, and $u_{HI}^{HI}=\frac{4}{5}$. 
%For \hi tasks $\tau_2$ and $\tau_3$, one can compute discriminant functions $\phi(\tau_2)=-\frac{1}{40}$ and $\phi(\tau_3)=-\frac{1}{10}$ in accordance with \defref{def:3}.
%A task system with three tasks is depicted in \tabref{ex:e1}. The mandatory service level $z_1^{man}$ for \lo tasks $\tau_1$ is $\frac{1}{8}$.  
%For the \hi tasks $\tau_2$ and $\tau_3$, $C_i^{LO}<C_i^{HI}$. From \tabref{ex:e1}, one can  derive $u_{LO}^{LO}=\frac{1}{2}$, $u_{LO}^{man}=\frac{1}{16}$, $u_{HI}^{LO}=\frac{1}{3}$, and $u_{HI}^{HI}=\frac{5}{8}$. For \hi tasks $\tau_2$ and $\tau_3$, one can compute discriminant functions $\phi(\tau_2)=-\frac{1}{40}$ and $\phi(\tau_3)=-\frac{1}{10}$ in accordance with \defref{def:3}.           
\end{example}

According to \thmref{thm:aa}, we can compute the uniform service levels $z^k$ for all possible mode-switching scenarios. The results are listed in \tabref{ex:e2}.

\begin{table}[h]
\centering
\caption{Low-criticality service levels}
\label{ex:e2}
\begin{tabular}{|c|c|c|c|c|}
\hline
Number of Overrun $k$    & $1$   & $2$    & $3$    & $4$  \\ \hline
Utilization $u_{LO}^{k}$ & 0.3   & 0.2    & 0.1    & 0  \\ \hline
Service Level $z^k$      & 0.75  & 0.5    & 0.25   & 0  \\ \hline
Execution Budget of $\tau_5$      & 22.5   & 15      & 7.5      & 0 \\ \hline
Execution Budget of $\tau_6$      & 56.25   & 37.5      & 18.75      & 0 \\ \hline
\end{tabular}
\end{table}

As shown in \tabref{ex:e2}, FMC-EDF-VD can efficiently support \lo task execution by dynamically tuning the \lo execution budget based on overrun demand. When only one \hi task overruns, \lo task $\tau_5$ and $\tau_6$ can use up to 22.5 and 56.25 time units per period for execution. In this case, \lo tasks can maintain 75\% execution. Only when all \hi tasks overrun their $C_i^{L}$, \lo tasks are all dropped. For comparison, the global triggering strategy used in~\cite{burns2013towards,DiLiu} are always required to drop all \lo tasks regardless of how many overruns occur during run time because of the overapproximation of the overrun workload. From a probabilistic perspective, the likelihood that all \hi tasks will exhibit \hi behavior is very low in practice. Therefore, in a typical case, only a few \hi tasks will overrun their $C_i^L$ during a busy interval. In most cases, FMC-EDF-VD will only need to schedule resources for a portion of \hi tasks based on their overrun demands and can maintain the service level for \lo task execution to the greatest possible extent. In this sense, FMC-EDF-VD can provide better and more graceful service degradation.
\vskip -3em 

\section{Schedulability Test Condition}
\label{sec:proof1}
\label{sec:correct}
In this section, we present a utilization-based schedulability test condition for the FMC-EDF-VD scheduling algorithm. We start by ensuring the schedulability of the system when it is operating in \lo mode (\thmref{theorem:lowcrit}). Then, we discuss how to derive a sufficient condition to ensure the schedulability of the algorithm after $k$ mode switches (\thmref{thm:3}). Based on several sophisticated new techniques, the correctness of this new schedulability test condition can be proven and the formal proof is provided in \secref{sec:corr}. Finally, we derive the region of $x$ that can guarantee the feasibility of the proposed scheduling algorithm.  
\subsection{Low-criticality mode}
In \lo mode, the system behaviors in FMC are exactly the same as in EDF-VD~\cite{Baruah2012}. Therefore, we can use the following theorem presented in \cite{Baruah2012} to ensure the schedulability of tasks in \lo mode. 

\begin{theorem}
\label{theorem:lowcrit}
 The following condition is sufficient to ensure that EDF-VD can successfully schedule all tasks in \lo mode:
 {\small
 \begin{equation}
  \label{equation:edvVD_low_bound}
u_{LO}^{LO}+\frac{u_{HI}^{LO}}{x}\le{1}
 \end{equation}
 }
\end{theorem}

\subsection{High-criticality mode after $k$ mode switches}
In this section, we analyze the schedulability of the FMC-EDF-VD algorithm during the transition phase. With this analysis, we provide the answer to the question of how much execution budget can be reserved for \lo tasks while ensuring a schedulable system for mode transitions. Without loss of generality, we consider a general transition case in which the system transitions from $(k-1)$-level \hi mode to $k$-level \hi mode. Here, we first introduce the derived schedulability test condition in \thmref{thm:3}. Then, the formal proof of the correctness of this schedulability test condition is provided in \secref{sec:corr}. Recall that $u_{LO}^{k}$ denotes the utilization of \lo tasks for the $k^{\textit{th}}$ mode-switching point and is defined as $u_{LO}^{k}=\sum_{\tau_i\in{\gamma_{LO}}}z_i^k\cdot{}u_i^{LO}$.   
 
\begin{theorem}
\label{thm:3}
The system is in $(k-1)$-level \hi mode. For the $k^{\textit{th}}$ mode-switching point $\hat{t}^k$, when \hi task $\tau_{\hat{t}^k}$ overruns, the system is schedulable at $\hat{t}^k$ if the following conditions are satisfied:
\begin{footnotesize}  
\begin{align}
\label{eq:thm4:m0}
u_{LO}^{k}&\le{}u_{LO}^{k-1}+\frac{\frac{u_{\hat{t}^k}^{LO}}{u_{HI}^{LO}}(1-u_{LO}^{LO})-u_{\hat{t}^k}^{HI}}{(1-x)} \\
\label{eq:thm4:m01}
z_i^k&\le{}z_i^{k-1}\ \ (\forall \tau_i \in \gamma_{LO})
\end{align}
\end{footnotesize}
where $u_{\hat{t}^k}^{LO}$ and $u_{\hat{t}^k}^{HI}$ denote low and high utilization, respectively, for the \hi task $\tau_{\hat{t}^k}$ that undergoes a mode switch at $\hat{t}^k$.    
\end{theorem}

In \thmref{thm:3}, we present a general utilization-based schedulability test condition for the FMC model. Now, let us take a closer look at the conditions specified in \thmref{thm:3}. We observe the following interesting properties of FMC-EDF-VD:
\begin{itemize}
\item In \thmref{thm:3}, the desired utilization balance between \lo and \hi tasks is achieved. As constrained by \eqnref{eq:thm4:m01}, the utilization of \lo tasks should be further reduced when a new overrun occurs. As shown in \eqnref{eq:thm4:m0}, the utilization reduction $u_{LO}^{k}-u_{LO}^{k-1}$ is bounded by $\frac{\frac{u_{\hat{t}^k}^{LO}}{u_{HI}^{LO}}(1-u_{LO}^{LO})-u_{\hat{t}^k}^{HI}}{(1-x)}$ for utilization balance.  
\item Another important observation is that the bound on the utilization reduction is determined \textit{only} by the overrun of \hi task $\tau_{\hat{t}^k}$ (as shown in \eqnref{eq:thm4:m0}). This means that the effects of the overruns on utilization reduction are independent. 
Moreover, the occurrence sequence of \hi task overruns has no impact on the utilization reduction. %As shown by the example presented in \tabref{ex:e2}, the overrun occurrence sequences $\tau_3\rightarrow\tau_2$ and $\tau_2\rightarrow\tau_3$ result in the same utilization reduction for \lo task $\tau_1$.   
\item \thmref{thm:3} also provides us with a generic metric for managing the resources of \lo tasks when each independent mode switch occurs. In general, various run-time tuning strategies can be applied during the transition phase, as long as the conditions in \thmref{eq:thm4:m0} are satisfied. 
%To compensate for the overload due to overruns, two strategies are considered: reducing the execution budgets for \lo tasks~\cite{burns2013towards} and partially dropping \lo tasks~\cite{gu,ren} (i.e., assigning $z_i^k=0$ for some proportion of \lo tasks). For both of these methods, \thmref{thm:3} provides the answers to the questions of the extent to which the execution budgets for \lo tasks must be reduced and how \lo tasks must be partially dropped to obtain a schedulable system depending on the run-time workload.   
\end{itemize}

\subsection{The proof of correctness}
\label{sec:corr}
We now prove the correctness of the schedulability test condition presented in \thmref{thm:3}. We start with the proof by introducing some important concepts. Then, we propose a key technique to obtain the bound of the cumulative execution time for \lo and \hi tasks (\lemref{lem:LOTask3}, \lemref{lem:hi0}, and \lemref{lem:hi1}). Based on these derived bounds, the utilization-based test condition can be derived.

\subsubsection{\textbf{Challenges}}
Incorporating the \FMC model into a utilization-based EDF-VD scheduling analysis introduces several new challenges. The independent triggering scheme and the adaptive service level tuning scheme in the \FMC model allow flexible system behaviors. However, this flexibility also makes the system behavior more complex and more difficult to analyze. In particular, it is difficult to \textit{effectively} determine an upper bound on the cumulative execution time for \lo tasks. In the \FMC model, the service levels for \lo tasks are dynamically tuned at each mode-switching point. Therefore, the cumulative execution time of \lo tasks strongly depends on when each mode switch occurs. However, this information is difficult to explicitly represent prior to real execution because the independence of the mode switches in the FMC model results in a large analysis space. This makes it computationally infeasible to analyze all possibilities. Moreover, apart from the timing information of multiple mode switches, the derivation of the cumulative execution time also depends on the service tuning decisions made at previous mode switches. Determining how to extract static information (i.e., utilization) to formulate a feasible sufficient condition from these variables is another challenging task.

\subsubsection{\textbf{Concepts and notation}}
Before diving into the detailed proofs, we introduce some commonly used concepts and notation that will be used throughout the proofs. To derive a sufficient test, suppose that there is a time interval $[0, t_f]$ such that the system undergoes the $k^{\textit{th}}$ mode switch and the first deadline miss occurs at $t_f$. Let $J$ be the minimal set of jobs released from the MC task set $\gamma$ for which a deadline is missed. This minimality means that if any job is removed from $J$, the remainder of $J$ will be schedulable. Here, we introduce some notation for later use. $\hat{t}^k$ denotes the time instant of the $k^{\textit{th}}$ mode switch caused by \hi task $\tau_{\hat{t}^k}$. The absolute release time and deadline of the job of $\tau_{\hat{t}^k}$ that overruns at $\hat{t}^k$ are denoted by $a_{\hat{t}^k}$ and $d_{\hat{t}^k}$, respectively. $\eta_i^k(t_1,t_2)$ denotes the cumulative execution time of task $\tau_i$ when the system is operating in $k$-level \hi mode during the interval $(t_1, t_2]$. Next, we define a special type of job for \lo tasks, called a \textit{carry-over job}, and introduce several important propositions that will be useful for our later proofs.
\begin{definition}
\label{definition:k_carry_over_job}
A job of \lo task $\tau_i$ is called a $k$-\textit{carry-over job} if the $k^{\text{th}}$ mode switch occurs in the interval $[a_i^k,d_i^k]$, where $a_i^k$ and $d_i^k$ are the absolute release time and deadline of this job, respectively.
\end{definition} 
\figref{fig:prop} shows how a $k$-\textit{carry-over job} is executed during the interval $[a_i^k,d_i^{k}]$. The black box represents the cumulative execution time $\eta_i^k(a_i^k,\hat{t}^{k})$ of the $k$-\textit{carry-over job} before the $k^{\text{th}}$ mode-switching point $\hat{t}^{k}$. 

\begin{figure}[t]
  \centering 
  \scalebox{1}{\begin{pspicture}(0,0)(8,1.9) 
%  \psgrid[subgriddiv=1,griddots=10,gridlabels=.2,gridcolor=lightgray](0,0)(8,1.5)
  \psset{linewidth=.8pt}%, framesep=.03pt}
  \scriptsize
  % \footnotesize
  \def \timer{\psline[linestyle=dashed,dash=0.8pt]{->}(0,0)(0,1)}
  \def\event{\psline[linewidth=1.2pt]{->}(0,0)(0,0.7)}
  \def\turnon{\psline[linewidth=1.2pt]{->}(0,0)(0,-1)}
	\def\pline{\psline[linewidth=0.8pt]{-}(0,0)(0,1)}
	\def\plinem{\psline[linewidth=0.8pt]{-}(0,0)(0,0.85)}
	\def\deadline{\psline[linewidth=1.2pt]{<-}(0,0)(0,0.7)}
	\def\rect1{\psframe*[linewidth=1.2pt,fillcolor=white](0,0)(0.5,0.4)}
	
   \rput(0,0.2){\pline}  \rput(0,0){$0$}
	 \rput(1.2,0.2){\plinem}     \rput(1.2,0){$\hat{t}^{k-j}$} 
   \rput(1.6,0.2){\event}     \rput(1.6,0){$a_i^k$} 
	 %\psline{<->}(1.6,0.5)(2.4,0.5) \rput(2,0.7){$z_{k-1}$}
	 %\rput(2.4,0.2){\plinem}     \rput(2.4,0){$\hat{t}^k$}
	 %\psline{<->}(2.4,0.5)(2.9,0.5) \rput(2.65,0.7){$z_{k}$}
   \rput(3.3,0.2){\plinem}     \rput(3.3,0){$\hat{t}^{k}$}   %\rput(3.4,0){$\cdots$} \rput(3.4,0.5){$\cdots$}
	 %\rput(4,0.2){\plinem}     \rput(4,0){$\hat{t}^{k+n_k}$}
	 %\psline{<->}(4,0.5)(5,0.5) \rput(4.5,0.7){$z_{k+n_k}$}
	 \rput(5,0.2){\deadline}     \rput(5,0){$d_{i}^k$}
   \rput(7.5,0.2){\pline}     \rput(7.5,0){$t_f$}
   
   %\psline(4.5,1.5)(4.5,1.9)
   %\psline{<->}(3.5,1.1)(7.5,1.1) \rput(5.5,1.3){\text{$z_1$}}
	 \rput(2.3,0.2){\rect1}     \rput(2.5,0.73){\fontsize{0.5 pt}{0.5 pt}\selectfont$\eta_i^k(a_i^k,\hat{t}^{k})$}
	 %\rput(3.8,0.2){\rect1}     \rput(4,0.72){$\Delta_2$}
	 %\psline(7.5,1.5)(7.5,1.9)
  \psline[linewidth=1pt]{->}(0,0.2)(8,0.2) % axis
  \rput[rb](8,0.3){t}
  
  % \psline[linewidth=3pt](.5,2)(2,2) \rput[t](1.25,2.5){1}
  % \psline[linewidth=3pt](3,2)(4.5,2) \rput[t](3.75,2.5){2}
  
%  \psline(2.5,.25)(2.5,2.5)   \psline(5,1.25)(5,2.5)
%  \psline(6,.75)(6,2.5)      \psline(8,.25)(8,2.5)
  
%  \psline{<->}(2.5,.6)(8,.6) \rput(5,.25){$sd^{max}_{1\,2}$}
%  \psline{<->}(2.5,1)(6,1) \rput(3.5,1.35){$sd^{min}_{1\,2}$}
%  \psline{<->}(5,1.8)(6,1.8) \rput(5.5,2.15){$sd^{min}_{1\,3}$}
%  \psline{<->}(5,1.4)(8,1.4) \rput(7,1.75){$sd^{max}_{1\,3}$}

\end{pspicture}

%%% Local Variables: 
%%% mode: latex
%%% TeX-master: "../paper"
%%% End: 
}
  \caption{The execution scenario for a $k$-\textit{carry-over job}.}
  \label{fig:prop}
\end{figure}
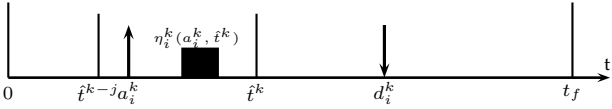 

\begin{proposition}
\label{prop:x}
(From \cite{Baruah2012,sanjoyACM}) 
All jobs executed in $[\hat{t}^k,t_f]$ have a deadline $\le{}t_f$.
\end{proposition}

\begin{proposition}
\label{prop:5}
The $k^{\textit{th}}$ mode-switching point $\hat{t}^k$ satisfies $\hat{t}^k\le a_{\hat{t}^k}+x\cdot{(t_f-a_{\hat{t}^k})}$.   
\end{proposition}
\begin{proof}
Since a \hi job of $\tau_{\hat{t}^{k}}$ triggers the $k^{\textit{th}}$ mode switch at $\hat{t}^{k}$, its virtual deadline ${a_{\hat{t}^{k}}+x\cdot{(d_{\hat{t}^{k}}-a_{\hat{t}^{k}})}}$ must be greater than $\hat{t}^{k}$. Otherwise, the \hi job would have completed its execution before the time instant of the switch.
\end{proof}

\begin{proposition}
\label{prop:x1}
For a $k$-\textit{carry-over job} of \lo task $\tau_i$, if $\eta_i^k(a_i^k,\hat{t}^{k})\neq{0}$, then the following holds: $d_i^k\le a_{\hat{t}^k}+x\cdot{(t_f-a_{\hat{t}^k})}$.  
\end{proposition}
\begin{proof}
There are two cases to consider: $a_i^k \ge a_{\hat{t}^k}$ and $a_i^k < a_{\hat{t}^k}$.

\noindent\textbf{Case 1} ($a_i^k \ge a_{\hat{t}^k}$):
In this case, for the $k$-\textit{carry-over job} to be executed after $a_{\hat{t}^k}$, 
the $k$-\textit{carry-over job} should have a deadline no later than the virtual deadline $a_{\hat{t}^k} +x(d_{\hat{t}^k}-a_{\hat{t}^k})$ of the \hi job that triggered the $k^{\text{th}}$ mode switch. 
As a result, because $d_{\hat{t}^k} \le t_f$, we have $d_i^k\le(a_{\hat{t}^k}+x\cdot{(t_f-a_{\hat{t}^k})})$.

\noindent\textbf{Case 2} ($a_i^k < a_{\hat{t}^k}$): We prove the correctness of this case by contradiction. Suppose that the $k$-\textit{carry-over job} of \lo task $\tau_i$, with its deadline of $d_i^k> (a_{\hat{t}^k}+x\cdot{(t_f-a_{\hat{t}^k})})$, were to be executed before $a_{\hat{t}^k}$. Let $t^*$ denote the latest 
time instant at which this $k$-\textit{carry-over job} is executed before $a_{\hat{t}^k}$. 
At time instant $\hat{t}^k$, all previous $(k-1)$ mode switches are known to the system\footnote{At $\hat{t}^k$, all previous $k-1$ mode switches have already occurred.}. 
At $t^*$, we know that there should be no pending job with a deadline of $\le (a_{\hat{t}^k}+x\cdot{(t_f-a_{\hat{t}^k})})$. This means that jobs that are released at or after $t^*$ will also suffer a deadline miss at $t_f$, which contradicts the minimality of $J$. Therefore, $d_i^k\le({a_{\hat{t}^k}+x\cdot{(t_f-a_{\hat{t}^k})}})$.
\end{proof}

Using the propositions and notation presented above, we now derive an upper bound on the cumulative execution time $\eta_i^k(0,t_f)$ for \lo tasks (\lemref{lem:LOTask3}) and \hi tasks (\lemref{lem:hi0} and \lemref{lem:hi1}).

\subsubsection{\textbf{Bound for \lo tasks}}
As discussed above, it is difficult to derive an upper bound on the cumulative execution time of \lo tasks during the interval $[0,t_f]$ because of the large analysis space. In this section, we propose a novel derivation strategy to resolve this challenge. The overall derivation strategy is based on the specified derivation protocol (\ruleref{rule:1}-\ruleref{rule:4}) and mathematical induction. The purpose of the derivation protocol is to specify unified \textbf{intermediate} upper bounds for different execution scenarios. The advantage of introducing these \textbf{intermediate} upper bounds is that we can \textit{virtually} hide the influence of the previous $k-1$ mode switches. For instance, in \ruleref{rule:1} (see \eqnref{eq:dp1}), the influence of the previous $k-1$ mode switches is hidden in the term $\sups{\eta_i^k(0,d_i^{l})}$. In this way, the $k^{\textit{th}}$ mode switch and the previous $k-1$ mode switches are decorrelated.

Throughout the remainder of this section, we will use $\sup{\{\eta_i^k(t_1,t_2)\}}$ to denote the \textbf{intermediate} upper bounds on $\eta_i^k(t_1, t_2)$ for different execution scenarios, which represent the upper bounds under specific conditions. Let $\hat{t}^{k-j}$ ($j>0$) denote the last mode-switching point before $a_i^k$ (as shown in \figref{fig:prop}). $z_i^{k-j}$ denotes the updated service level at $\hat{t}^{k-j}$. $d_i^{l}$ denotes the absolute deadline for the last job\footnote{Here, the last job means the last job with a deadline of $\leq{t_f}$.} of $\tau_i$ during $[0,t_f]$. Now, we present the rules for deriving $\sups{\eta_i^k(0,t_f)}$ and $\sups{\eta_i^k(a_i^k,d_i^{k})}$, as summarized in \eqnref{eq:dp1} and \eqnref{eq:dp2}.

\begin{small}
\begin{align}
\label{eq:dp1}
&\sups{\eta_i^k(0,t_f)} \nonumber\\
=&
\begin{cases} 
\sups{\eta_i^k(0,d_i^{l})}+(t_f-\hat{t}^k)\cdot{}z_i^k\cdot{}u_i^{LO} &\text{$d_i^{l}<\hat{t}^k$ (\textbf{\ruleref{rule:1}})}\\
\sups{\eta_i^k(0,d_i^{k})}+(t_f-d_i^{k})\cdot{}z_i^k\cdot{}u_i^{LO} &\text{Otherwise (\textbf{\ruleref{rule:2}})} \\
\end{cases}\\ 
\label{eq:dp2}
&\sups{\eta_i^k(a_i^k,d_i^{k})} \nonumber\\
=&
\begin{cases} 
{(d_i^k-a_i^k)\cdot{z_i^{k-j}}\cdot{u_i^{LO}}} &\text{$\eta_i^k(a_i^k,\hat{t}^{k})\neq{0}$ (\textbf{\ruleref{rule:3}})}\\
{(d_i^k-a_i^k)\cdot{z_i^{k}}\cdot{u_i^{LO}}} &\text{Otherwise (\textbf{\ruleref{rule:4}})}
\end{cases}
\end{align}
\end{small}

The detailed description and proof are presented in Appendix \ref{appendix:II}. In \ruleref{rule:1}-\ruleref{rule:4}, one may notice that there are several different execution scenarios in which only one mode switch is considered. When $n$ mode switches are allowed, the combination space for all execution scenarios will increase exponentially with $n$. In general, it is very difficult to derive a bound on the cumulative execution time for \lo tasks because of this large combination space. To solve this problem, we analyze the difference between $\sups{\eta_i^{k}(0,t_f)}$ and $\sups{\eta_i^{k-1}(0,t_f)}$ and find that this difference can be uniformly bounded by a \textit{difference term} $\psi_i^{k}$ (see \lemref{lem:LOTask3}). This finding is formally proven in \lemref{lem:LOTask3} through mathematical induction. Before the proof, we first present a fact that will be useful for later interpretation.
\begin{fact}
\label{rule:5}
For the $k^{\textit{th}}$ mode-switching point $\hat{t}^{k}$, at time instant $t_0$ such that $t_0\le{\hat{t}^{k}}$, $\eta_i^k(0,t_0)=\eta_i^{k-1}(0,t_0)$.
%$k^{\textit{th}}$ mode-switch has no impact on $\eta_i^k(0,t_0)$.  
\end{fact}
\begin{proof}
The $k^{\textit{th}}$ mode switch can only begin to affect \lo task execution after the corresponding mode-switching point $\hat{t}^{k}$. Before $\hat{t}^{k}$, the $k^{\textit{th}}$ mode switch has no impact. Thus,  we have $\eta_i^k(0,t_0)=\eta_i^{k-1}(0,t_0)$.    
\end{proof}

\begin{lemma}
\label{lem:LOTask3}
For all $k\ge{1}$, the cumulative execution time $\eta_i^{k}(0,t_f)$ can be upper bounded by
\begin{small}
\begin{align}
t_f\cdot{u_i^{LO}}+\sum_{j=1}^{k}\psi_i^{j}
\label{eq:LOTask3:1}
\end{align}
\end{small}
where the \textit{difference term} $\psi_i^{j}$ is defined as $(t_f-a_{\hat{t}^{j}})(1-x)(z_i^{j}-z_i^{j-1})u_i^{LO}$.   
\end{lemma}
\begin{proof}
Instead of proving the original statement, we will prove an alternative statement $P(k)$, which is defined as follows:

\textit{The \textbf{intermediate} upper bounds $\sups{\eta_i^{k}(0,t_f)}$ under different execution scenarios can be uniformly upper bounded by \eqnref{eq:LOTask3:1}.} 

Since $\eta_i^{k}(0,t_f)\le{}\sups{\eta_i^{k}(0,t_f)}$, the original statement will be proven correct if the statement $P(k)$ is proven to be correct. Now, we will prove that the statement $P(k)$ is correct for all possible integers $k$ based on mathematical induction. Recall that  $d_i^{l}$ is the absolute deadline for the last job of $\tau_i$ during $[0,t_f]$.  

\noindent\textbf{Step 1 (base case)}: We will prove that $P(1)$ is correct for $k=1$. 

\begin{proof}
We consider two cases, one in which a \textit{carry-over job} does not exist at the first mode-switching point $\hat{t}^1$ (i.e., $d_i^{l}<\hat{t}^1$) and one in which such a job does exist (i.e., $d_i^{l}\ge{}\hat{t}^1$).

\textbf{Case 1 ($\mathbf{d_i^{l}<\hat{t}^{1}}$)}: According to \ruleref{rule:1} and \propref{prop:5}, we have the following:
\begin{small}
\begin{align*}
&\sups{\eta_i^1(0,t_f)} \\
&=\sups{\eta_i^1(0,d_i^{l})}+(t_f-\hat{t}^1)\cdot{}z_i^1\cdot{}u_i^{LO} \\
&=d_i^{l}\cdot{}u_i^{LO}+(t_f-\hat{t}^1)\cdot{}z_i^1\cdot{}u_i^{LO} \\
&since\ \ d_i^{l}<\hat{t}^1\le{}a_{\hat{t}^1}+x\cdot{(t_f-a_{\hat{t}^1})} \\
&<{}t_f\cdot{}z_i^1\cdot{}u_i^{LO}+\hat{t}^1\cdot{}u_i^{LO}\cdot{}(1-z_i^1)\ \ \ \ \  (replace\ \ d_i^{l}\ \ with\ \ \hat{t}^1) \\
&\le{}t_f\cdot{}u_i^{LO}+\underbrace{(t_f-a_{\hat{t}^1})(1-x)(z_i^1-1)u_i^{LO}}_{difference\;term\;\psi_i^{1}} \ \ \  (replace\ \ \hat{t}^1) 
\end{align*} 
\end{small}

\textbf{Case 2 ($\mathbf{d_i^{l}\ge{}\hat{t}^{1}}$)}: In this case, we consider the two following execution scenarios.

\noindent\textbf{S1 ($\mathbf{\eta_i^1(a_i^{1},\hat{t}^{1})\neq{0}})$}:
According to \ruleref{rule:2},  \ruleref{rule:3}, and \propref{prop:x1}, we have the following:
\begin{small}
\begin{align*}
	&\sups{\eta_i^{1}(0,t_f)} \nonumber \\
=&\sups{\eta_i^{1}(0,a_i^{1})}+\sups{\eta_i^{1}(a_i^{1},d_i^{1})}+(t_f-d_i^{1}){z_i^{1}}{u_i^{LO}} \nonumber  \\
=&a_i^{1}{u_i^{LO}}+(d_i^{1}-a_i^{1}){u_i^{LO}}+(t_f-d_i^{1}){z_i^{1}}{u_i^{LO}}\nonumber \\
=&t_f\cdot{}u_i^{LO}+(t_f-d_i^{1})(z_i^1-1)u_i^{LO} \nonumber \\
&since\  d_i^{1}\le{a_{\hat{t}^1}+x\cdot{(t_f-a_{\hat{t}^1})}} \nonumber \\
\le{}& t_f\cdot{}u_i^{LO}+\underbrace{(t_f-a_{\hat{t}^1})(1-x)(z_i^1-1)u_i^{LO}}_{difference\;term\;\psi_i^{1}}  \ \ \  (replace\ \ d_i^{1}) 
\end{align*}
\end{small}

\noindent\textbf{S2 ($\mathbf{\eta_i^1(a_i^1,\hat{t}^{1})={0}}$)}: 
According to  \ruleref{rule:2},  \ruleref{rule:4}, and \propref{prop:5}, we have the following:
\begingroup
\allowdisplaybreaks
\begin{small}
\begin{align*}
	&\sups{\eta_i^{1}(0,t_f)} \nonumber \\
=&\sups{\eta_i^{1}(0,a_i^{1})}+\sups{\eta_i^{1}(a_i^{1},d_i^{1})}+(t_f-d_i^{1}){z_i^{1}}{u_i^{LO}} \nonumber  \\
=&a_i^{1}{u_i^{LO}}+(d_i^{1}-a_i^{1})z_i^1{u_i^{LO}}+(t_f-d_i^{1}){z_i^{1}}{u_i^{LO}}\nonumber \\
=&t_f\cdot{}u_i^{LO}+(t_f-a_{i}^1)(z_i^1-1)u_i^{LO} \nonumber \\
&since\  a_{i}^1<\hat{t}^1\le{a_{\hat{t}^1}+x\cdot{(t_f-a_{\hat{t}^1})}} \nonumber \\
\le{}& t_f\cdot{}u_i^{LO}+\underbrace{(t_f-a_{\hat{t}^1})(1-x)(z_i^1-1)u_i^{LO}}_{difference\;term\;\psi_i^{1}}   \ \ \  (replace\ \ a_{i}^1) 
\end{align*}
\end{small}
\endgroup
Therefore, $P(1)$ is correct for $k=1$. 
\end{proof}

\noindent\textbf{Step 2 (induction hypothesis)}: Assume that $P(k_0-1)$ is correct for some possible integers $k_0-1$.

\noindent\textbf{Step 3 (induction)}: We now prove that $P(k_0)$ is correct by the induction hypothesis.

\begin{proof}
Since $\hat{t}^{k_0-1}\le{}\hat{t}^{k_0}$, we need to consider the following three cases.

\textbf{Case 1 ($\mathbf{d_i^{l}<\hat{t}^{k_0-1}\le{}\hat{t}^{k_0}}$)}: In this case, neither a $(k_0-1)$-\textit{carry-over job} nor a $k_0$-\textit{carry-over job} exists.  
According to \ruleref{rule:1} and \factref{rule:5}, we have the following:
\begin{small}
\begin{align*}
\sups{\eta_i^{k_0-1}(0,t_f)}&=\sups{\eta_i^{k_0-1}(0,d_i^{l})}+(t_f-\hat{t}^{k_0-1})z_i^{k_0-1}u_i^{LO} \\
\sups{\eta_i^{k_0}(0,t_f)}&=\sups{\eta_i^{k_0}(0,d_i^{l})}+(t_f-\hat{t}^{k_0})z_i^{k_0}u_i^{LO} \\
\sups{\eta_i^{k_0-1}(0,d_i^{l})}&=\sups{\eta_i^{k_0}(0,d_i^{l})}
\end{align*} 
\end{small}

Since $\hat{t}^{k_0}\ge{}\hat{t}^{k_0-1}$ and $z_i^{k_0}\le{}z_i^{k_0-1}$, we have
\begin{small}
\begin{align}
\sups{\eta_i^{k_0}(0,t_f)}\le{}\sups{\eta_i^{k_0-1}(0,t_f)}+(t_f-\hat{t}^{k_0}){(z_i^{k_0}-z_i^{k_0-1})}{u_i^{LO}}
\label{eq:lem1:2}
\end{align} 
\end{small}
According to \propref{prop:5}, we can replace $\hat{t}^{k_0}$ with $a_{\hat{t}^{k_0}}+x{(t_f-a_{\hat{t}^{k_0}})}$ in \eqnref{eq:lem1:2}. Then, $\sups{\eta_i^{k_0}(0,t_f)}$ can be bounded by
\begin{small}
\begin{align*}
\sups{\eta_i^{k_0-1}(0,t_f)}+\underbrace{(t_f-a_{\hat{t}^{k_0}})\cdot(1-x)\cdot{(z_i^{k_0}-z_i^{k_0-1})}\cdot{u_i^{LO}}}_{difference\;term\;\psi_i^{k_0}}
\end{align*} 
\end{small} 

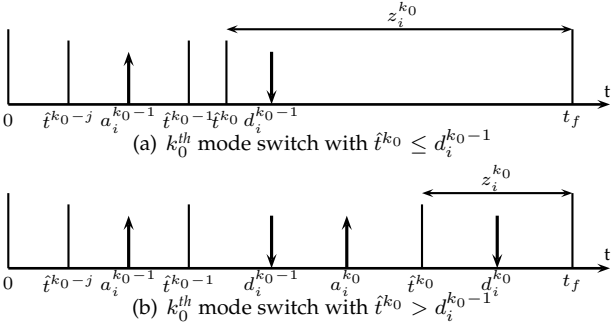
\begin{figure}
\centering
  \subfigure[$k_0^\textit{th}$ mode switch with $\hat{t}^{k_0}\le{d_i^{k_0-1}}$]
  {\label{fig:sec2} \scalebox{1}{\begin{pspicture}(0,0)(8,1.5) 
%  \psgrid[subgriddiv=1,griddots=10,gridlabels=.2,gridcolor=lightgray](0,0)(8,1.5)
  \psset{linewidth=.8pt}%, framesep=.03pt}
  \scriptsize
  % \footnotesize
  \def \timer{\psline[linestyle=dashed,dash=0.8pt]{->}(0,0)(0,1)}
  \def\event{\psline[linewidth=1.2pt]{->}(0,0)(0,0.7)}
  \def\turnon{\psline[linewidth=1.2pt]{->}(0,0)(0,-1)}
	\def\pline{\psline[linewidth=0.8pt]{-}(0,0)(0,1)}
	\def\plinem{\psline[linewidth=0.8pt]{-}(0,0)(0,0.85)}
	\def\deadline{\psline[linewidth=1.2pt]{<-}(0,0)(0,0.7)}
	\def\rect1{\psframe*[linewidth=1.2pt,fillcolor=white](0,0)(0.4,0.4)}
	
   \rput(0,0.2){\pline}  \rput(0,0){$0$}
   \rput(1.6,0.2){\event}     \rput(1.6,0){$a_i^{k_0-1}$} 
	 \rput(0.8,0.2){\plinem}     \rput(0.8,0){$\hat{t}^{k_0-j}$}
	 \rput(2.4,0.2){\plinem}     \rput(2.4,0){$\hat{t}^{k_0-1}$}
	 \psline{<->}(2.9,1.2)(7.5,1.2) \rput(5.25,1.4){$z_i^{k_0}$}
   \rput(2.9,0.2){\plinem}     \rput(2.9,0){$\hat{t}^{k_0}$}   
	 \rput(3.5,0.2){\deadline}     \rput(3.5,0){$d_{i}^{k_0-1}$}
	 %\psline{<->}(4,0.5)(5,0.5) \rput(4.5,0.7){$z_{k+n_k}$}
	 %\rput(5,0.2){\deadline}     \rput(5,0){$d_{i}^k$}
   \rput(7.5,0.2){\pline}     \rput(7.5,0){$t_f$}
   
   %\psline(4.5,1.5)(4.5,1.9)
   %\psline{<->}(3.5,1.1)(7.5,1.1) \rput(5.5,1.3){\text{$z_1$}}
	 %\rput(2.8,0.2){\rect1}     \rput(3,0.72){$\Delta_1$}
	 %\rput(3.8,0.2){\rect1}     \rput(4,0.72){$\Delta_2$}
	 %\psline(7.5,1.5)(7.5,1.9)
  \psline[linewidth=1pt]{->}(0,0.2)(8,0.2) % axis
  \rput[rb](8,0.3){t}
  
  % \psline[linewidth=3pt](.5,2)(2,2) \rput[t](1.25,2.5){1}
  % \psline[linewidth=3pt](3,2)(4.5,2) \rput[t](3.75,2.5){2}
  
%  \psline(2.5,.25)(2.5,2.5)   \psline(5,1.25)(5,2.5)
%  \psline(6,.75)(6,2.5)      \psline(8,.25)(8,2.5)
  
%  \psline{<->}(2.5,.6)(8,.6) \rput(5,.25){$sd^{max}_{1\,2}$}
%  \psline{<->}(2.5,1)(6,1) \rput(3.5,1.35){$sd^{min}_{1\,2}$}
%  \psline{<->}(5,1.8)(6,1.8) \rput(5.5,2.15){$sd^{min}_{1\,3}$}
%  \psline{<->}(5,1.4)(8,1.4) \rput(7,1.75){$sd^{max}_{1\,3}$}

\end{pspicture}}
  }
  \subfigure[$k_0^\textit{th}$ mode switch with $\hat{t}^{k_0}>{d_i^{k_0-1}}$]
  {\label{fig:sec3}   \scalebox{1}{\begin{pspicture}(0,0)(8,1.5) 
%  \psgrid[subgriddiv=1,griddots=10,gridlabels=.2,gridcolor=lightgray](0,0)(8,1.5)
  \psset{linewidth=.8pt}%, framesep=.03pt}
  \scriptsize
  % \footnotesize
  \def \timer{\psline[linestyle=dashed,dash=0.8pt]{->}(0,0)(0,1)}
  \def\event{\psline[linewidth=1.2pt]{->}(0,0)(0,0.7)}
  \def\turnon{\psline[linewidth=1.2pt]{->}(0,0)(0,-1)}
	\def\pline{\psline[linewidth=0.8pt]{-}(0,0)(0,1)}
	\def\plinem{\psline[linewidth=0.8pt]{-}(0,0)(0,0.85)}
	\def\deadline{\psline[linewidth=1.2pt]{<-}(0,0)(0,0.7)}
	\def\rect1{\psframe*[linewidth=1.2pt,fillcolor=white](0,0)(0.4,0.4)}
	
   \rput(0,0.2){\pline}  \rput(0,0){$0$}
   \rput(1.6,0.2){\event}     \rput(1.6,0){$a_i^{k_0-1}$} 
	 \rput(0.8,0.2){\plinem}     \rput(0.8,0){$\hat{t}^{k_0-j}$}
	 \rput(2.4,0.2){\plinem}     \rput(2.4,0){$\hat{t}^{k_0-1}$}
	 \psline{<->}(5.5,1.2)(7.5,1.2) \rput(6.5,1.4){$z_i^{k_0}$}  
	 \rput(3.5,0.2){\deadline}     \rput(3.5,0){$d_{i}^{k_0-1}$}
	
	 \rput(4.5,0.2){\event}     \rput(4.5,0){$a_i^{k_0}$}
	 \rput(5.5,0.2){\plinem}     \rput(5.5,0){$\hat{t}^{k_0}$} 
	 \rput(6.5,0.2){\deadline}     \rput(6.5,0){$d_{i}^{k_0}$}
	 %\psline{<->}(4,0.5)(5,0.5) \rput(4.5,0.7){$z_{k+n_k}$}
	 %\rput(5,0.2){\deadline}     \rput(5,0){$d_{i}^k$}
   \rput(7.5,0.2){\pline}     \rput(7.5,0){$t_f$}
   
   %\psline(4.5,1.5)(4.5,1.9)
   %\psline{<->}(3.5,1.1)(7.5,1.1) \rput(5.5,1.3){\text{$z_1$}}
	 %\rput(2.8,0.2){\rect1}     \rput(3,0.72){$\Delta_1$}
	 %\rput(3.8,0.2){\rect1}     \rput(4,0.72){$\Delta_2$}
	 %\psline(7.5,1.5)(7.5,1.9)
  \psline[linewidth=1pt]{->}(0,0.2)(8,0.2) % axis
  \rput[rb](8,0.3){t}
  
  % \psline[linewidth=3pt](.5,2)(2,2) \rput[t](1.25,2.5){1}
  % \psline[linewidth=3pt](3,2)(4.5,2) \rput[t](3.75,2.5){2}
  
%  \psline(2.5,.25)(2.5,2.5)   \psline(5,1.25)(5,2.5)
%  \psline(6,.75)(6,2.5)      \psline(8,.25)(8,2.5)
  
%  \psline{<->}(2.5,.6)(8,.6) \rput(5,.25){$sd^{max}_{1\,2}$}
%  \psline{<->}(2.5,1)(6,1) \rput(3.5,1.35){$sd^{min}_{1\,2}$}
%  \psline{<->}(5,1.8)(6,1.8) \rput(5.5,2.15){$sd^{min}_{1\,3}$}
%  \psline{<->}(5,1.4)(8,1.4) \rput(7,1.75){$sd^{max}_{1\,3}$}

\end{pspicture}}
  }
  \caption{Mode switch from $\hat{t}^{k_0-1}$ to $\hat{t}^{k_0}$.}
  \label{fig:timing}
	\vspace{-0.75cm}
\end{figure}

\textbf{Case 2 ($\mathbf{{\hat{t}^{k_0-1}\le{}\hat{t}^{k_0}\le{}d_i^{l}}}$)}: In this case, both a $(k_0-1)$-\textit{carry-over job} and a $k_0$-\textit{carry-over job} exist. Recall that $d_i^{k_0-1}$ is the absolute deadline for the $(k_0-1)$-\textit{carry-over job}. Two sub-cases, one with $\hat{t}^{k_0}\le{d_i^{k_0-1}}$ and one with $\hat{t}^{k_0}>{d_i^{k_0-1}}$, as shown in \figref{fig:sec2} and \figref{fig:sec3}, need to be considered.

According to  \factref{rule:5}, we have the following:
\begin{small}
\begin{align}
\label{eq:lema_diff1_2}
\eta_i^{k_0}(0,a_i^{k_0})=\eta_i^{k_0-1}(0,a_i^{k_0})
\end{align}
\end{small}

\noindent$\bullet$\textbf{Case 2-A ($\hat{t}^{k_0}\le{d_i^{k_0-1}}$)}: This execution scenario is illustrated in \figref{fig:sec2}. In this case, the $(k_0-1)$-\textit{carry-over job} and the $k_0$-\textit{carry-over job} are the same job. Therefore, we have $a_i^{k_0}=a_i^{k_0-1}$ and $d_i^{k_0}=d_i^{k_0-1}$. In the following, we use $a_i^{k_0-1}$ and $d_i^{k_0-1}$ in place of $a_i^{k_0}$ and $d_i^{k_0}$, respectively. In Case 2-A, the following two scenarios are considered: 

\noindent\textbf{S1 ($\mathbf{\eta_i^{k_0}(a_i^{{k_0}-1},\hat{t}^{k_0})\neq{0}})$}: 
According to  \ruleref{rule:3} and \ruleref{rule:4}, we have the following\footnote{According to the proof of \ruleref{rule:3} (see Appendix \ref{appendix:II}), we have a similar result: $\sups{\eta_i^{k_0}(a_i^{k_0-1},d_i^{k_0-1})}=(d_i^{k_0-1}-a_i^{k_0-1})\cdot{z_i^{k_0-1}}\cdot{u_i^{LO}}$ because $\eta_i^k(a_i^{k_0-1},\hat{t}^{k_0-1})={0}$.}:  
\begin{small}
\begin{align}
\label{eq:lema_diff1_1}
&\sups{\eta_i^{k_0}(a_i^{k_0-1},d_i^{k_0-1})}=\sups{\eta_i^{k_0-1}(a_i^{k_0-1},d_i^{k_0-1})} \nonumber \\
=& \begin{cases} 
   (d_i^{k_0-1}-a_i^{k_0-1})\cdot{z_i^{k_0-j}}\cdot{u_i^{LO}} & \eta_i^k(a_i^{k_0-1},\hat{t}^{k_0-1})\neq{0} \\
   (d_i^{k_0-1}-a_i^{k_0-1})\cdot{z_i^{k_0-1}}\cdot{u_i^{LO}} & otherwise
  \end{cases}
\end{align}
\end{small}

According to  \ruleref{rule:2}, \eqnref{eq:lema_diff1_2} and \eqnref{eq:lema_diff1_1},  we have the following:
\begin{small}
\begin{align*}
	&\sups{\eta_i^{k_0}(0,t_f)} \nonumber \\
=&\sups{\eta_i^{k_0}(0,a_i^{k_0-1})}+\sups{\eta_i^{k_0}(a_i^{k_0-1},d_i^{k_0-1})} \nonumber  \\
&+(t_f-d_i^{k_0-1}){z_i^{k_0}}{u_i^{LO}} \nonumber  \\
=&\underline{\sups{\eta_i^{k_0-1}(0,a_i^{k_0-1})}+\sups{\eta_i^{k_0-1}(a_i^{k_0-1},d_i^{k_0-1})}} \nonumber \\
 &\underline{+(t_f-d_i^{k_0-1}){z_i^{k_0-1}}{u_i^{LO}}}+(t_f-d_i^{k_0-1}){(z_i^{k_0}-z_i^{k_0-1})}{u_i^{LO}} \nonumber \\
=&\sups{\eta_i^{k_0-1}(0,t_f)}+(t_f-d_i^{k_0-1}){(z_i^{k_0}-z_i^{k_0-1})}{u_i^{LO}}
\end{align*}
\end{small}

According to \propref{prop:x1}, by replacing $d_i^{k_0-1}$ with $a_{\hat{t}^{k_0}}+x\cdot{(t_f-a_{\hat{t}^{k_0}})}$, $\sups{\eta_i^{k_0}(0,t_f)}$ can be bounded by
\begin{small}
\begin{align}
\sups{\eta_i^{k_0-1}(0,t_f)}+\underbrace{(t_f-a_{\hat{t}^{k_0}})\cdot(1-x)\cdot{(z_i^{k_0}-z_i^{k_0-1})}\cdot{u_i^{LO}}}_{difference\;term\;\psi_i^{k_0}}
\label{eq:lema_diff1_3}
\end{align} 
\end{small} 

\noindent\textbf{S2 ($\mathbf{\eta_i^{k_0}(a_i^{k_0-1},\hat{t}^{k_0})={0}}$)}:
According to  \ruleref{rule:2},  \ruleref{rule:4}, and \eqnref{eq:lema_diff1_2},  we have the following:
\begin{small}
\begin{align*}
	&\sups{\eta_i^{k_0}(0,t_f)} \nonumber \\
=&\sups{\eta_i^{k_0}(0,a_i^{k_0-1})}+\sups{\eta_i^{k_0}(a_i^{k_0-1},d_i^{k_0-1})} \nonumber  \\
&+(t_f-d_i^{k_0-1})z_i^{k_0}u_i^{LO} \nonumber  \\
=&\underline{\sups{\eta_i^{k_0-1}(0,a_i^{k_0-1})}+\sups{\eta_i^{k_0-1}(a_i^{k_0-1},d_i^{k_0-1})}} \nonumber \\
 &\underline{+(t_f-d_i^{k_0-1})z_i^{k_0-1}u_i^{LO}}+(t_f-a_i^{k_0-1}){(z_i^{k_0}-z_i^{k_0-1})}{u_i^{LO}} \nonumber \\
=&\sups{\eta_i^{k_0-1}(0,t_f)}+(t_f-a_i^{k_0-1}){(z_i^{k_0}-z_i^{k_0-1})}{u_i^{LO}}
\end{align*}
\end{small}

According to \propref{prop:5} and $a_i^{k_0-1}<\hat{t}^{k_0}$, $\sups{\eta_i^{k_0}(0,t_f)}$ can be bounded by
\begin{small}
\begin{align}
\sups{\eta_i^{k_0-1}(0,t_f)}+\underbrace{(t_f-a_{\hat{t}^{k_0}})(1-x)(z_i^{k_0}-z_i^{k_0-1})u_i^{LO}}_{difference\;term\;\psi_i^{k_0}}
\label{eq:lema_diff1_4}
\end{align}
\end{small} 

\noindent$\bullet$\textbf{Case 2-B: (${d_i^{k_0-1}}<\hat{t}^{k_0}$)}: This execution scenario is illustrated in \figref{fig:sec3}. In this case, the $(k_0-1)$-\textit{carry-over job} and the $k_0$-\textit{carry-over job} are different jobs. For this case, we will consider the following two scenarios:

\noindent\textbf{S1 ($\mathbf{\eta_i^{k_0}(a_i^{k_0},\hat{t}^{k_0})\neq{0}}$)}: 
According to  \ruleref{rule:2},  \ruleref{rule:3}, and \eqnref{eq:lema_diff1_2}, we have the following: 
\begin{footnotesize}
\begin{align*}
	&\sups{\eta_i^{k_0}(0,t_f)} \nonumber \\
=&\sups{\eta_i^{k_0}(0,a_i^{k_0})}+\sups{\eta_i^{k_0}(a_i^{k_0},d_i^{k_0})}+(t_f-d_i^{k_0})z_i^{k_0}u_i^{LO} \nonumber  \\
=&\underline{\sups{\eta_i^{k_0-1}(0,a_i^{k_0})}+(t_f-a_i^{k_0}){z_i^{k_0-1}}{u_i^{LO}}} \nonumber \\
&+(t_f-d_i^{k_0}){(z_i^{k_0}-z_i^{k_0-1})}{u_i^{LO}} \nonumber \\
=&\sups{\eta_i^{k_0-1}(0,t_f)}+(t_f-d_i^{k_0}){(z_i^{k_0}-z_i^{k_0-1})}{u_i^{LO}}
\end{align*}
\end{footnotesize}

Again, by replacing $d_i^{k_0}$ in accordance with \propref{prop:x1}, we obtain the following bound:
\begin{small}
\begin{align}
\sups{\eta_i^{k_0-1}(0,t_f)}+\underbrace{(t_f-a_{\hat{t}^{k_0}})(1-x)(z_i^{k_0}-z_i^{k_0-1}){u_i^{LO}}}_{difference\;term\;\psi_i^{k_0}}
\label{eq:lema_diff1_5}
\end{align}
\end{small}

\noindent\textbf{S2 ($\mathbf{\eta_i^{k_0}(a_i^{k_0},\hat{t}^{k_0})={0}}$)}:
According to  \ruleref{rule:2},  \ruleref{rule:4}, and \eqnref{eq:lema_diff1_2}, we have the following: 
\begin{footnotesize}
\begin{align*}
	&\sups{\eta_i^{k_0}(0,t_f)} \nonumber \\
=&\sups{\eta_i^{k_0}(0,a_i^{k_0})}+\sups{\eta_i^{k_0}(a_i^{k_0},d_i^{k_0})}+(t_f-d_i^{k_0})z_i^{k_0}u_i^{LO} \nonumber  \\
=&\underline{\sups{\eta_i^{k_0-1}(0,a_i^{k_0})}+(t_f-a_i^{k_0}){z_i^{k_0-1}}{u_i^{LO}}} \nonumber \\
&+(t_f-a_i^{k_0}){(z_i^{k_0}-z_i^{k_0-1})}{u_i^{LO}} \nonumber \\
=&\sups{\eta_i^{k_0-1}(0,t_f)}+(t_f-a_i^{k_0}){(z_i^{k_0}-z_i^{k_0-1})}{u_i^{LO}}
\end{align*}
\end{footnotesize}

Again, according to \proporef{prop:5} and $a_i^{k_0}<\hat{t}^{k_0}$, $\sups{\eta_i^{k}(0,t_f)}$ can be upper bounded by
\begin{small}
\begin{align}
\sups{\eta_i^{k_0-1}(0,t_f)}+\underbrace{(t_f-a_{\hat{t}^{k_0}})(1-x)(z_i^{k_0}-z_i^{k_0-1}){u_i^{LO}}}_{difference\;term\;\psi_i^{k_0}}
\label{eq:lema_diff1_6}
\end{align}
\end{small}

For case 2, we can conclude that $\sups{\eta_i^{k}(0,t_f)}$ can be upper bounded by $\sups{\eta_i^{k_0-1}(0,t_f)}+\psi_i^{k_0}$ according to \eqnsref{eq:lema_diff1_3}-(\ref{eq:lema_diff1_6}). 

\textbf{Case 3 (${\hat{t}^{k_0-1}\le{}d_i^{l}<\hat{t}^{k_0}}$)}: In this case, a $(k_0-1)$-\textit{carry-over job} exists but a $k_0$-\textit{carry-over job} does not. 
According to  \ruleref{rule:1}, we have the following:
\begin{small}
\begin{align*}
\sups{\eta_i^{k_0}(0,t_f)}&=\sups{\eta_i^{k_0}(0,d_i^{l})}+(t_f-\hat{t}^{k_0})\cdot{}z_i^{k_0}\cdot{}u_i^{LO}
\end{align*} 
\end{small}
Since $d_i^{l}<\hat{t}^{k_0}<t_f$, we can derive
\begin{small}
\begin{align*}
\sups{\eta_i^{k_0-1}(0,t_f)}=\sups{\eta_i^{k_0-1}(0,d_i^{l})}+(t_f-d_i^{l})\cdot{}z_i^{k_0-1}\cdot{}u_i^{LO}
\end{align*} 
\end{small}
According to  \factref{rule:5} and $d_i^{l}<\hat{t}^{k_0}$, we have
\begin{small}
\begin{align*}
\sups{\eta_i^{k_0}(0,t_f)}\le{}\sups{\eta_i^{k_0-1}(0,t_f)}+(t_f-\hat{t}^{k_0})(z_i^{k_0}-z_i^{k_0-1})u_i^{LO} 
\end{align*} 
\end{small}

Again, according to \proporef{prop:5}, $\sups{\eta_i^{k}(0,t_f)}$ can be upper bounded by
\begin{small}
\begin{align*}
\sups{\eta_i^{k_0-1}(0,t_f)}+\underbrace{(t_f-a_{\hat{t}^{k_0}})(1-x)(z_i^{k_0}-z_i^{k_0-1}){u_i^{LO}}}_{difference\;term\;\psi_i^{k_0}}
\end{align*}
\end{small}

For the three cases above, we can conclude that $\sups{\eta_i^{k}(0,t_f)}$ can be upper bounded by $\sups{\eta_i^{k_0-1}(0,t_f)}+\psi_i^{k_0}$. Thus, $P(k_0)$ is correct by the induction hypothesis. 
\end{proof}
Hence, through mathematical induction, $P(k)$ is proven correct for all possible $k$. Under different execution scenarios, the cumulative execution time $\eta_i^{k}(0,t_f)$ can be bounded by the \textbf{intermediate} upper bound $\sups{\eta_i^{k}(0,t_f)}$. Since $P(k)$ is correct, the original statement is correct.   
\end{proof}

\subsubsection{\textbf{Bound for \hi tasks}}
Recall that $\tau_{\hat{t}^k}$ is the \hi task that suffers an overrun at $\hat{t}^k$. Since the mode switches are independent, the \hi tasks can be divided into two sets, namely, the sets of tasks that have and have not already entered \hi mode at mode-switching point $\hat{t}^k$, which can be denoted by $\gamma_{HI}^{HI}(\hat{t}^k)$ and $\gamma_{HI}^{LO}(\hat{t}^k)$, respectively. Now, we derive the upper bounds on the cumulative execution time for both types of \hi tasks.         

\begin{lemma}
\label{lem:hi0}
For \hi task $\tau_{\hat{t}^j}$ in task set $\gamma_{HI}^{HI}(\hat{t}^k)$ ($j\le{k}$), the cumulative execution time $\eta_{\tau_{\hat{t}^j}}^k(0,t_f)$ can be bounded as follows:
\begin{small}
\begin{align}
\label{eq:hi0}
\sups{\eta_{\tau_{\hat{t}^j}}^k(0,t_f)} = a_{\hat{t}^j}\cdot{}u_{\hat{t}^j}^{LO}+(t_f-a_{\hat{t}^j})\cdot{}u_{\hat{t}^j}^{HI}
\end{align}
\end{small}
\end{lemma}
\begin{proof}
For the proof, refer to case 2 of fact 3 in \cite{Baruah2012}. 
\end{proof}

\begin{lemma}
\label{lem:hi1}
For \hi task $\tau_i$ in task set $\gamma_{HI}^{LO}(\hat{t}^k)$, the cumulative execution time $\eta_{i}^k(0,t_f)$ can be bounded as follows:
\begin{small}
\begin{align}
\label{eq:hi1}
\sups{\eta_{i}^k(0,t_f)} = (\frac{a_{\hat{t}^k}}{x}+(t_f-a_{\hat{t}^k}))u_{i}^{LO}
\end{align}
\end{small}
\end{lemma}
\begin{proof}
For the proof, refer to the fact 3 in \cite{Baruah2012}. 
\end{proof}

\subsubsection{\textbf{Putting it all together}}
\label{sec:pia}
Now, we are ready to establish the schedulability test condition. To prove \thmref{thm:3}, we first introduce two auxiliary theorems, \thmref{thm:1} and \thmref{thm:2}. In \thmref{thm:1}, the schedulability test condition is derived based on \lemref{lem:LOTask3}, \lemref{lem:hi0}, and \lemref{lem:hi1}. This test condition should rely on the previous mode switches. \thmref{thm:2} demonstrates the consistency of the test condition, by which the dependences among mode switches can be removed.     

\begin{theorem}
\label{thm:1}
At the $k$-th mode-switching point $\hat{t}^{k}$, $k$ ($k\ge{}1$) \hi tasks $\tau_{\hat{t}^{1}},\tau_{\hat{t}^{2}},\cdots,\tau_{\hat{t}^{k}}$ have switched into \hi mode. 
The system is schedulable if the service level $z_i^j$ at $\hat{t}^{j}$ satisfies the following conditions for \textbf{all} $j$ such that $1\le{j}\le{k}$. 
\begin{footnotesize}  
\begin{align}
\label{eq:thm3:m3}
&z_i^{j}\le{z_i^{j-1}}\\
\label{eq:thm3:m2}
&u_{\hat{t}^j}^{HI}+(1-x)(u_{LO}^{j}-u_{LO}^{j-1})+\frac{u_{\hat{t}^j}^{LO}}{u_{HI}^{LO}}(u_{LO}^{LO}-1)\le{0}
\end{align} 
\end{footnotesize}
\end{theorem}

\begin{proof}
The condition $z_i^{j}\le{z_i^{j-1}}$ is a basic assumption of our model, which guarantees the satisfaction of \lemref{lem:LOTask3}, \ruleref{rule:3}, and \ruleref{rule:4}. Therefore, $z_i^{j}\le{z_i^{j-1}}$ needs to be satisfied.

Let $N^{k}_{\gamma}$ denote the cumulative execution time of task set $\gamma$ during the interval $[0,t_f]$ when the $k^{\textit{th}}$ mode switch occurs. To calculate $N^{k}_{\gamma}$, let us sum the the cumulative execution time of all tasks over $[0,t_f]$. 

For the \lo task set $\gamma_{LO}$, we can bound $N^{k}_{\gamma_{LO}}$ according to \lemref{lem:LOTask3}.
\begin{small}
\begin{align}
N^{k}_{\gamma_{LO}}\le{}\sum_{\tau_i\in\gamma_{LO}}(t_fu_i^{LO}+\sum_{j=1}^{k}{\psi_i^{j}})
\label{eq:1}
\end{align}
\end{small}

For the \hi task set $\gamma_{HI}^{HI}(\hat{t}^{k})$, which contains $k$ \hi tasks (i.e., $\|\gamma_{HI}^{HI}(\hat{t}^{k})\|=k$), we can derive the cumulative execution time according to \lemref{lem:hi0}.
\begin{small}
\begin{align}
N^{k}_{\gamma_{HI}^{HI}(\hat{t}^{k})}\le{}\sum_{j=1}^{k}{(a_{\hat{t}^j}\cdot{}u_{\hat{t}^j}^{LO}+(t_f-a_{\hat{t}^j})\cdot{}u_{\hat{t}^j}^{HI})}
\label{eq:2}
\end{align}
\end{small}

For the \hi tasks in $\gamma_{HI}^{LO}(\hat{t}^{k})$, which have not entered \hi mode at $\hat{t}^{k}$, we can derive the cumulative execution time according to \lemref{lem:hi1}.  
\begin{small}
\begin{align}
N^{k}_{\gamma_{HI}^{LO}(\hat{t}^{k})}&\le{}\sum_{\tau_i\in\gamma_{HI}^{LO}(\hat{t}^{k})}{(\frac{a_{\hat{t}^k}}{x}+(t_f-a_{\hat{t}^k}))u_{i}^{LO}}  \nonumber \\
&\ (since \ x<1\ and\ a_{\hat{t}^k}\le{t_f}) \nonumber\\
&\le{}\sum_{\tau_i\in\gamma_{HI}^{LO}(\hat{t}^{k})}{\frac{t_f}{x}u_{i}^{LO}} 
\label{eq:3}
\end{align}
\end{small}

Based on \eqnref{eq:1}, \eqnref{eq:2} and \eqnref{eq:3}, $N^{k}_{\gamma}$ can be bounded as shown in \eqnref{eq:main:0}. The complete derivation is given in Appendix \ref{appendix:I} because of space limitations. 
\begingroup
\allowdisplaybreaks
\begin{scriptsize}
\begin{align}
&N^{k}_{\gamma}=N^{k}_{\gamma_{LO}}+N^{k}_{\gamma_{HI}^{HI}(\hat{t}^{k})}+N^{k}_{\gamma_{HI}^{LO}(\hat{t}^{k})} \nonumber \\
\le&t_f+\sum_{j=1}^{k}(t_f-a_{\hat{t}^j})\bigg(u_{\hat{t}^j}^{HI}+(1-x)(u_{LO}^{j}-u_{LO}^{j-1})+\frac{u_{\hat{t}^j}^{LO}}{u_{HI}^{LO}}(u_{LO}^{LO}-1)\bigg) 
\label{eq:main:0}
\end{align}
\end{scriptsize}
\endgroup
Since the first deadline miss occurs at time instant $t_f$, the following holds\footnote{Note that there is no idle instant within the interval $[0,t_f]$. 
Otherwise, jobs from set $J$ with release times at or after the latest idle instant could form a smaller job set causing a deadline miss at $t_f$, which would contradict the minimality of $J$.}:
\begin{small}
\begin{align*}
N^{k}_{\gamma}>{t_f}
\end{align*}
\end{small}

Therefore, 
\begin{footnotesize}
\begin{align*}
&\sum_{j=1}^{k}(t_f-a_{\hat{t}^j})\bigg(u_{\hat{t}^j}^{HI}+(1-x)(u_{LO}^{j}-u_{LO}^{j-1})+\frac{u_{\hat{t}^j}^{LO}}{u_{HI}^{LO}}(u_{LO}^{LO}-1)\bigg)>{0}
\end{align*}
\end{footnotesize}

Taking the contrapositive, we obtain
\begin{footnotesize}
\begin{align}
&\sum_{j=1}^{k}(t_f-a_{\hat{t}^j})\underline{\bigg(u_{\hat{t}^j}^{HI}+(1-x)(u_{LO}^{j}-u_{LO}^{j-1})+\frac{u_{\hat{t}^j}^{LO}}{u_{HI}^{LO}}(u_{LO}^{LO}-1)\bigg)}\le{0}
\label{eq:4}
\end{align}
\end{footnotesize}

Since $t_f-a_{\hat{t}^j}>0$, to guarantee the system schedulability of task set $\gamma$ at the $k^{\textit{th}}$ mode switch, it is sufficient to ensure that the term indicated in \eqnref{eq:4} is less than 0 for all $j$ such that $1\le{j}\le{k}$.
\begin{footnotesize}
\begin{align}
&\forall j\ such\ that\ 1\le{j}\le{k}: \nonumber \\
&u_{\hat{t}^j}^{HI}+(1-x)(u_{LO}^{j}-u_{LO}^{j-1})+\frac{u_{\hat{t}^j}^{LO}}{u_{HI}^{LO}}(u_{LO}^{LO}-1)\le{0}
\label{eq:5}
\end{align}
\end{footnotesize}    
\end{proof}

In \thmref{thm:1}, at the $k^{\textit{th}}$ mode-switching point, additional conditions are imposed on the previous $k-1$ mode switches. Therefore, to remove this dependence, we require that these imposed conditions should be consistent with the decision-making at the previous mode-switching points $\hat{t}^j$ ($j<k$). We demonstrate this consistency in \thmref{thm:2}.   

\begin{theorem}
\label{thm:2}   
The new conditions imposed on $u_{LO}^{1},u_{LO}^{2},\cdots,u_{LO}^{k-1}$ by the $k^{\textit{th}}$ mode switch are consistent with the decisions that have been made at the previous mode-switching points.      
\end{theorem}
\begin{proof}
The conditions given in \thmref{thm:1} for decisions that have been made at the previous $k-1$ mode-switching points $\hat{t}^{j}$ ($1\le{}j\le{k-1}$) are exactly the same as the new conditions imposed on $u_{LO}^{1},u_{LO}^{2},\cdots,u_{LO}^{k-1}$ with the $k^{\textit{th}}$ mode switch. Therefore, their consistency is guaranteed.     
\end{proof}

\noindent\textbf{\FMC schedulability}: Now, we are ready to prove \thmref{thm:3} using \thmref{thm:1} and \thmref{thm:2}.
 
\begin{proof}
According to \thmref{thm:2}, the constraints in \thmref{thm:1} that are imposed on  $u_{LO}^{1},u_{LO}^{2},\cdots,u_{LO}^{k-1}$ with the $k^{\textit{th}}$ mode switch have already been covered by the previous $k-1$ mode switches. Therefore, we need to check only two conditions: \eqnref{eq:thm3:m3} and \eqnref{eq:thm3:m2} with $j=k$.
\end{proof}

%%%%%%%%%%%%%%%%%%%%%%%%%%%%%%%%%%%%%%%%%%%%%%%%%%%%%%%%%%%%%%%%%%%%%%%%%%%%%%%%%%%%%%%%%%%%%%%

\subsection{Feasibility of Algorithm}
\label{sec:fea}
In this section, we investigate the region of $x$ values that can guarantee the feasibility of the run-time algorithm. The selection of any $x$ from this region during the off-line phase can guarantee that a feasible solution as determined by \thmref{thm:3} can always be found during run time. To derive this region, we first introduce several definitions and properties that will be useful for the later proof of feasibility. 

According to \eqnref{eq:thm4:m0} in \thmref{thm:3}, when $\frac{u_{\hat{t}^k}^{LO}}{u_{HI}^{LO}}(1-u_{LO}^{LO})-u_{\hat{t}^k}^{HI}>{0}$, we do not need to reduce the utilization of \lo tasks. The overrun of the \hi task at this mode-switching point is covered by the system resource margin. Only when $\frac{u_{\hat{t}^k}^{LO}}{u_{HI}^{LO}}(1-u_{LO}^{LO})-u_{\hat{t}^k}^{HI}\le{0}$, $u_{LO}^{k}$ should be decreased to compensate for the overrun of the \hi task. For simplicity, we define a discriminant function $\phi(\tau_i)$ for each \hi task $\tau_i$ to indicate whether the overrun of $\tau_i$ can be covered by the system resource margin.    
\begin{definition}
\label{def:3}
 $\phi(\tau_i)=\frac{u_{i}^{LO}}{u_{HI}^{LO}}(1-u_{LO}^{LO})-u_{i}^{HI}\ \ \ (\tau_i\in\gamma_{HI})$
\end{definition}  

\begin{definition}
\label{def:4}
A \hi task $\tau_i$ is called margin \hi task if $\phi(\tau_i)>0$. Otherwise, $\tau_i$ is called compensation  \hi task.    
\end{definition}

\begin{definition}
\label{def:5}
The margin \hi task set and the compensation \hi task set are defined as $\gamma_{HI}^{\circ}=\{\tau_i\in\gamma_{HI}|\phi(\tau_i)>{0}\}$ and $\gamma_{HI}^{\ast}=\{\tau_i\in\gamma_{HI}|\phi(\tau_i)\le{0}\}$, respectively. $\gamma_{HI}= \gamma_{HI}^{\circ} \cup \gamma_{HI}^{\ast}$.    
\end{definition}   

With the definitions given above, we can now perform the feasibility analysis for $x$.  
\begin{theorem}
\label{thm:5}
Given the mandatory utilization $u_{LO}^{man}$, any $x$ that satisfies the following condition can guarantee that a feasible solution as determined by \thmref{thm:3} can always be found during run time.  
\begin{footnotesize}  
\begin{align}
\label{eq:thm5:m0}
(1-x)(u_{LO}^{LO}-u_{LO}^{man})+\sum\limits_{\tau_{i}\in{\gamma_{HI}^{\ast}}}\phi(\tau_i)\ge{}0
\end{align}
\end{footnotesize} 
\end{theorem}

\begin{proof}
Recall that $\gamma_{HI}^{HI}(\hat{t}^{k})$ is the set of \hi tasks that have entered \hi mode at $\hat{t}^{k}$. By iterating the conditions in \thmref{thm:3}, a direct solution for $u_{LO}^{k}$ can be obtained as follows: 
\begin{footnotesize}  
\begin{align}
u_{LO}^{k}\le{u_{LO}^{LO}}+\frac{\sum\limits_{\tau_{i}\in{\gamma_{HI}^{\ast}\cap{\gamma_{HI}^{HI}(\hat{t}^{k})}}}\phi(\tau_i)}{(1-x)}
\label{eq:5:1}
\end{align}
\end{footnotesize}

To guarantee the execution of the mandatory portions of \lo tasks, the following condition should be satisfied for all $k$:
\begin{footnotesize}  
\begin{align}
u_{LO}^{man}\le{}u_{LO}^{k}\le{u_{LO}^{LO}}+\frac{\sum\limits_{\tau_{i}\in{\gamma_{HI}^{\ast}\cap{\gamma_{HI}^{HI}(\hat{t}^{k})}}}\phi(\tau_i)}{(1-x)}
\label{eq:5:0}
\end{align}
\end{footnotesize}   
Since the right-hand side of \eqnref{eq:5:0} is non-increasing with respect to the number of overrun \hi tasks (i.e., $k$), the worst-case scenario is that all \hi tasks in $\gamma_{HI}$ enter \hi mode. If mandatory service can be guaranteed in this worst-case scenario, then the feasibility of the proposed algorithm is ensured. Therefore, condition \eqnref{eq:5:0} can be rewritten as \eqnref{eq:thm5:15}.      
\begin{footnotesize}  
\begin{align}
&u_{LO}^{LO}+\frac{\sum\limits_{\tau_{i}\in{\gamma_{HI}^{\ast}}}\phi(\tau_i)}{(1-x)}\ge{}u_{LO}^{man} \nonumber \\
\Rightarrow\ \ \ &(1-x)(u_{LO}^{LO}-u_{LO}^{man})+\sum\limits_{\tau_{i}\in{\gamma_{HI}^{\ast}}}\phi(\tau_i)\ge{}0 
\label{eq:thm5:15}
\end{align}
\end{footnotesize}
{
Note that $u_{LO}^{man}$ is the mandatory utilization defined as $u_{LO}^{man}=\sum_{\tau_i\in{\gamma_{LO}}}z_{i}^{man}\cdot{}u_i^{LO}$, where the item $z_{i}^{man}\cdot{}u_i^{LO}$ can be considered as a mandatory part which affects the correctness of the result in imprecise computation model~\cite{LinIMC}.
}
\vskip -1em   
\end{proof}
{
Now, we use the following example to illustrate how to test the feasibility of FMC-EDF-VD. }
\begin{example}
{
Considering the task system in Example \ref{example1}, we can derive $x=\frac{u_{HI}^{LO}}{1-u_{LO}^{LO}}=\frac{1}{2}$ according to \thmref{theorem:lowcrit}. For \hi tasks, one can compute discriminant functions $\phi(\tau_1)=\phi(\tau_2)=\phi(\tau_3)=\phi(\tau_4)=-\frac{1}{20}$ in accordance with \defref{def:3}. The feasibility of $x$ is validated by checking condition \eqnref{eq:thm5:m0} in \thmref{thm:5}.
\begin{small}
\begin{align}
&(1-x)(u_{LO}^{LO}-u_{LO}^{man})+\sum\limits_{\tau_{i}\in{\gamma_{HI}^{\ast}}}\phi(\tau_i) \nonumber\\
=&(1-\frac{1}{2})\times{}\frac{2}{5}-4\times{}\frac{1}{20}=0 \nonumber
\end{align}
\end{small}
Thus, we know $x=\frac{1}{2}$ that is feasible for scheduling using FMC-EDF-VD.}
%For clarity of presentation, we consider a task set that contains four identical \hi tasks and four identical \lo tasks, as listed in \tabref{ex:e3}. We specify $u_{LO}^{man}=0$ for demonstration. From \tabref{ex:e3}, one can  derive $u_{LO}^{LO}=\frac{2}{5}$, $u_{HI}^{LO}=\frac{3}{10}$, and $u_{HI}^{HI}=\frac{4}{5}$.       
\end{example}
\vskip -3em
%%%%%%%%%%%%%%%%%%%%%%%%%%%%%%%%%%%%%%%%%%%%%%%%%%%%%%%
\section{Service level tuning strategy}
\label{sec:slts}
\thmref{thm:3} provides an important criterion for run-time service level tuning. By checking the conditions in \thmref{thm:3}, one can determine how much utilization can be reserved for \lo task execution to compensate for the overruns. In general, various tuning strategies can be specified by the user as long as the condition in \thmref{thm:3} is satisfied during run time. In this paper, we present a uniform tuning strategy and a dropping-off strategy to demonstrate the performance of FMC.     
\subsection{Dropping-off strategy}
To compensate for overruns, the dropping-off strategy partially drops \lo tasks by assigning $z_i^k=0$ for dropped tasks. To maximize the utilization of \lo tasks, the tasks to be dropped can be selected according to their utilization. At each mode-switching point $\hat{t}^k$, tasks with less utilization are given higher priority for dropping. To implement this selection strategy, we can create a task table $TA_{LO}$ during the off-line phase by sorting the \lo tasks in ascending order of their utilization. During run time, 
the utilization reduction $U_{R}^k$ that is required to compensate for the $k^\textit{th}$ mode switch is determined according to \thmref{thm:3}. Based on $TA_{LO}$, the set $\gamma_{LO}^k$ of tasks that are dropped at the $k^\textit{th}$ mode-switching point is determined via binary search. Note that other selection criteria, such as job completion percentage, can also be applied to select the \lo tasks to be dropped.

\subsection{Uniform tuning strategy}
In this section, we present a uniform tuning strategy in which $z_i^k=z^k$ holds for all \lo tasks. The service levels $z_i^k$ of all \lo tasks $\tau_i$ are uniformly set to $z^k$ at the $k^{\textit{th}}$ mode-switching point. By applying $z_i^k=z^k$ in the conditions given in \thmref{thm:3}, the uniform service level $z^k$ can be directly computed using \eqnref{eq:thm4:m011} in \thmref{thm:aa}. 

\begin{theorem}
\label{thm:aa}
The system is schedulable at the $k-1^{\textit{th}}$ mode-switching point with a uniform $z^{k-1}$. At the $k^{\textit{th}}$ mode-switching point $\hat{t}^k$, the system is still schedulable if $z^k$ is determined as follows:
\begin{footnotesize}  
\begin{align}
\label{eq:thm4:m011}
0\le{}z^{k}\le{z^{k-1}}+\min{}\Big(0,\frac{\frac{u_{\hat{t}^k}^{LO}}{u_{HI}^{LO}}(1-u_{LO}^{LO})-u_{\hat{t}^k}^{HI}}{(1-x)u_{LO}^{LO}}\Big)
\end{align}
\end{footnotesize}
where $u_{\hat{t}^k}^{LO}$ and $u_{\hat{t}^k}^{HI}$ denote low and high utilization, respectively, of the \hi task $\tau_{\hat{t}^k}$ that suffers an overrun at $\hat{t}^k$.    
\end{theorem}
\begin{proof}
In the uniform tuning strategy, $z_i^k=z^k$ holds for any \lo task $\tau_i$. Recall that $u_{LO}^{k}=\sum_{\tau_i\in{\gamma_{LO}}}z_i^k\cdot{}u_i^{LO}$. Thus, we can obtain \eqnref{eq:thm4:m011} by combining the two conditions expressed in \eqnref{eq:thm4:m0} and \eqnref{eq:thm4:m01} in \thmref{thm:3}. 
\vskip -1em   
\end{proof}
\vskip -3em

\subsection{Case study}
\label{sec:cs}
{
In this case study, we firstly use the task system in Example \ref{example1} to illustrate how uniform tuning strategy and dropping-off strategy work in FMC. Then, we implement the uniform tuning strategy in our simulation framework (presented in Appendix \ref{appendix:III}) to demonstrate the graceful \lo service degradation of FMC.
  
First of all, we consider the generalized conditions presented in \thmref{thm:3}, which determine how much utilization can be reserved for \lo task execution to compensate for the overruns. By applying the task system presented in Example \ref{example1} to \thmref{thm:3}, we can the following utilization conditions：
\begin{align}
\label{example:gen:1}
u_{LO}^{k}-u_{LO}^{k-1}&\le{\frac{\frac{u_{\hat{t}^k}^{LO}}{u_{HI}^{LO}}(1-u_{LO}^{LO})-u_{\hat{t}^k}^{HI}}{(1-x)}}=-\frac{1}{10} \\
\label{example:gen:2}
z_i^k&\le{}z_i^{k-1}\ \ (\forall \tau_i \in \gamma_{LO})
\end{align}
Since the \hi tasks are identical, each overrun will result in identical utilization reduction of $\frac{1}{10}$, as shown in \eqnref{example:gen:1}. Now, we illustrate how uniform tuning strategy and dropping-off strategy work based on these generalized conditions \eqnref{example:gen:1} and \eqnref{example:gen:2}. 
\begin{itemize}
\item For dropping-off strategy by assigning $z_i^k=0$ for dropped tasks, the system is required to drop off a portion of \lo tasks to compensate for the overruns of one \hi task. For example, when one \hi task overruns its $C_i^L$, \lo task $\tau_5$ may decrease its execution budget from 30 to 10, while \lo task $\tau_6$ is executed without degradation. By this way, the service degradation of $\tau_5$ results in utilization reduction of $\frac{1}{10}$ to accommodate one \hi overrun. The dropping-off process is summarized in \tabref{ex:e4}.       
%
%the system is required to drop-off one \lo task to compensate the overrun of one \hi task, because dropping-off one \lo task can result in $\frac{1}{10}$ utilization reduction. Therefore, the system in $k$-level \hi mode can accommodate $4-k$ \lo tasks for execution. 
\item For uniform tuning strategy by restricting $z_i^k=z^k$ for all \lo tasks, each overrun will result in an identical reduction of 0.25 in $z_k$, such that the condition \eqnref{example:gen:1} is satisfied. Therefore, the service level $z_k$ for operation in $k$-level \hi mode can be expressed as $z_k=1-0.25\cdot{}k$.  
\item By contrast, if one were to apply IMC~\cite{DiLiu} to this task set, the guaranteed service level would be 0. This means that any overrun would result in the dropping off of all \lo tasks.   
\end{itemize}

\begin{table}[h]
\centering
\caption{Low-criticality service levels}
\label{ex:e4}
\begin{tabular}{|c|c|c|c|c|}
\hline
Number of Overrun $k$    & $1$   & $2$    & $3$    & $4$  \\ \hline
Utilization $u_{LO}^{k}$ & 0.3   & 0.2    & 0.1    & 0  \\ \hline
Execution Budget of $\tau_5$      & 10   & 0      & 0      & 0 \\ \hline
Execution Budget of $\tau_6$      & 75   & 60      & 30      & 0 \\ \hline
\end{tabular}
\end{table}
}

Next, we evaluate the implementation of the FMC-EDF-VD run-time system in our simulation framework to demonstrate the graceful \lo service degradation of FMC. In this case study, the uniform tuning strategy is applied for demonstration. We ran the simulation for $2\times{}10^6$ time units, which contains $5\times{}10^4$ \hi jobs. We set the \hi job behavior probability to 0.1. The simulation process is detailed in Appendix \ref{appendix:III}.

%The run-time service level for each \lo job was collected from our simulation environment. 
%To validate the schedulability of FMC in the worst-case scenario, we assumed all jobs ran with their worst-case execution times for stress testing. 
\figref{fig:gd} shows the run-time service levels for both FMC and IMC~\cite{DiLiu}.
The lower bounds on the service levels with different numbers of mode switches, as discussed above, for FMC and IMC are represented by red and black lines, respectively, in \figref{fig:gd}. The dashed green line represents service level $z^{k-1}$ for operation in $(k-1)$-level \hi mode for FMC. The collected run-time service levels as scheduled by FMC are represented in the form of box-whisker plots with blue dots.      
\begin{figure}
\centering
\includegraphics[width=0.8\columnwidth,height=0.42\columnwidth]{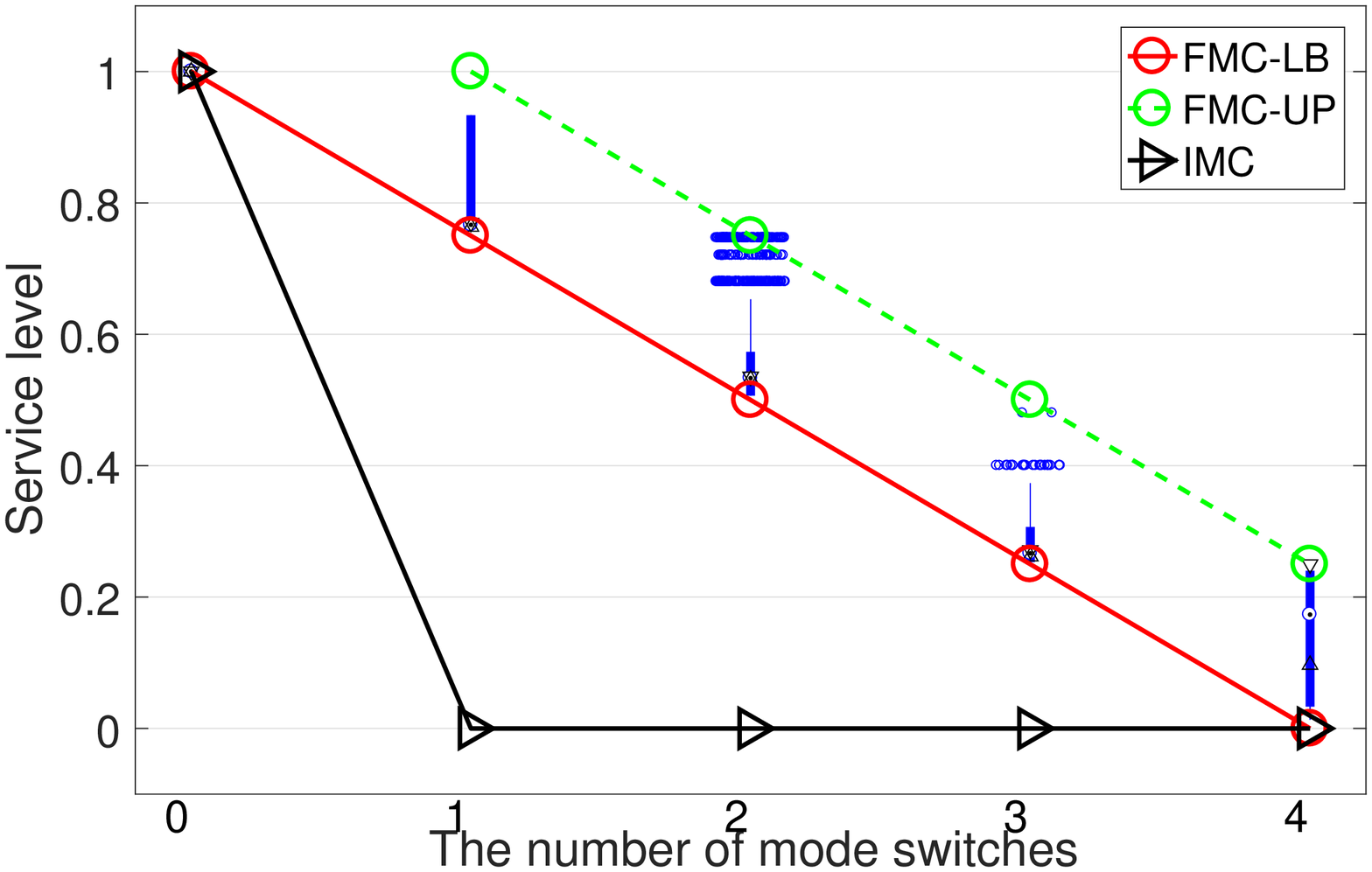}
 \caption{Service level of \lo tasks under the different number of mode switches.}
  \label{fig:gd}
\end{figure}

As shown in \figref{fig:gd}, FMC can gracefully degrade the \lo service level as the number of mode switches increases. By contrast, IMC fails to respond to the variability in the workload. As long as not all \hi tasks overrun during run time, the execution budget determined by FMC always outperforms that of IMC. 

Another interesting observation is that the collected run-time service levels are bounded by the red and green lines. This observation matches the FMC execution semantics presented in \secref{sec:sem}. In the $k^{\textit{th}}$ transition phase, the execution budget for \lo jobs will be reduced from $z^{k-1}\cdot{C_i^{LO}}$ to $z^{k}\cdot{C_i^{LO}}$. According to the FMC execution semantics, two cases can be considered:
\begin{itemize}
\item \textbf{Case 1:} Low-criticality jobs that have already exhausted their execution budget of $z^{k}\cdot{C_i^{LO}}$ at the transition point. Such jobs will be suspended immediately. In addition, these suspended jobs should have an execution time of less than $z^{k-1}\cdot{C_i^{LO}}$ by the $k$-th transition point. Otherwise, these jobs would have already been suspended when the system entered $(k-1)-level$ \hi mode. Therefore, the execution time of these jobs will be bounded in $[z^{k}\cdot{C_i^{LO}},z^{k-1}\cdot{C_i^{LO}})$. 
\item \textbf{Case 2:} Low-criticality jobs that have not yet exhausted their execution budget $z^{k}\cdot{C_i^{LO}}$. Such jobs will continue to run until their remaining time budget is used up. Therefore, these jobs will execute up to $z^{k}\cdot{C_i^{LO}}$.
\end{itemize}
From the above two cases, we can conclude that the execution time of these jobs in $k$-level \hi mode is bounded in $[z^{k}{C_i^{LO}},z^{k-1}{C_i^{LO}})$, as clearly shown in \figref{fig:gd}.

\subsection{Run-time complexity}
According to~\cite{Baruah2012}, for a task set containing $n$ tasks, the classic EDF-VD algorithm has a run-time complexity of $O(log \,n)$ per event for job arrival, job completion, and mode switching. Compared with EDF-VD~\cite{Baruah2012}, FMC-EDF-VD needs to implement only one additional operation during mode switching, that is, tuning the service levels for \lo tasks according to the specified strategy. For the uniform tuning strategy, the uniform service level $z_k$ can be directly computed with a complexity of $O(1)$ according to \thmref{thm:aa}. For the dropping-off strategy, the dropping-off task can be determined via binary search with a complexity of $O(log\,n)$. Therefore, FMC-EDF-VD still has a run-time complexity of $O(log\,n)$ per event. 

%%%%%%%%%%%%%%%%%%%%%%%%%%%%%%%%%%%%%%%

\vskip -2em
\section{Evaluation}
\label{sec:eva}
In this section, simulation experiments are presented to evaluate the performance of FMC. Our experiments are based on randomly generated MC tasks. We randomly generate task sets using the same approach as in \cite{Baruah2012,gu}. The various parameters are set as follows:
\begin{itemize}
\item The period $T_i$ of each task is an integer drawn uniformly at random from $[20,150]$.
\item For each task $\tau_i$, \lo utilization $u_i^{LO}$ is a real number drawn at random from $[0.05, 0.15]$.
\item $R_i$ denotes the ratio of $u_i^{HI}/u_i^{LO}$, which is a real number drawn uniformly at random from $[2,3]$.
\item pCri denotes the probability that a task $\tau_i$ is a \hi task, and we set this probability to $0.5$. If $\tau_i$ is a \lo task, then we set $C_i^{LO}=\floor{u_i^{LO}\cdot{T_i}}$. Otherwise, we set $C_i^{LO}=\floor{u_i^{LO}\cdot{T_i}}$ and $C_i^{HI}=\floor{u_i^{LO}\cdot{}R_i\cdot{}{T_i}}$. 
\end{itemize}
Given the utilization bound $u_B$, we generate one task at a time until the following conditions are both satisfied: (1) $u_B -0.05 \le \max\{u_{LO}^{LO} + u_{HI}^{LO},u_{HI}^{HI}\} \le u_B$. (2) At least 3 \hi tasks have been generated.

The generated task set is evaluated for both off-line schedulability and on-line performance in terms of support for \lo task execution under six different schemes. These schemes include FMC with dropping off strategy proposed in this paper('FMC'), Pfair-based  scheme with using task grouping~\cite{ren}('PF'), component-based scheme~\cite{gu}('COM'), advanced EDF-VD scheduling of IMC systems~\cite{DiLiu}('IMC'), service adaption strategy that decreases the dispatch frequency of \lo tasks based on EDF-VD scheduling~\cite{HGST14a}('SA'), classic EDF-VD scheduling~\cite{Baruah2012}('EDF-VD'). 

The on-line \lo performance is measured as the percentage of finished LC jobs (denoted by PFJ), which is the same quantitative parameter used in~\cite{gu}. PFJ is defined as the percentage of \lo jobs that are successfully finished by their deadlines. Each simulation is run for $10^6$ time units. The execution distribution presented in \cite{Back} is used to compute the probability that a \hi task $\tau_i$ will be executed beyond its \lo WCET. {Due to schedulability performance differences among the compared schemes, the PFJ is obtained only when the taskset is schedulable for all compared schemes. The simulation process is detailed in Appendix \ref{appendix:III}.} 

%To ensure fair comparisons, we generate a job trace for each generated task set and use this unified job trace to obtain PFJ for all compared schemes. 
\vskip -1em
\begin{figure}
\centering
\includegraphics[width=0.8\columnwidth,height=0.4\columnwidth]{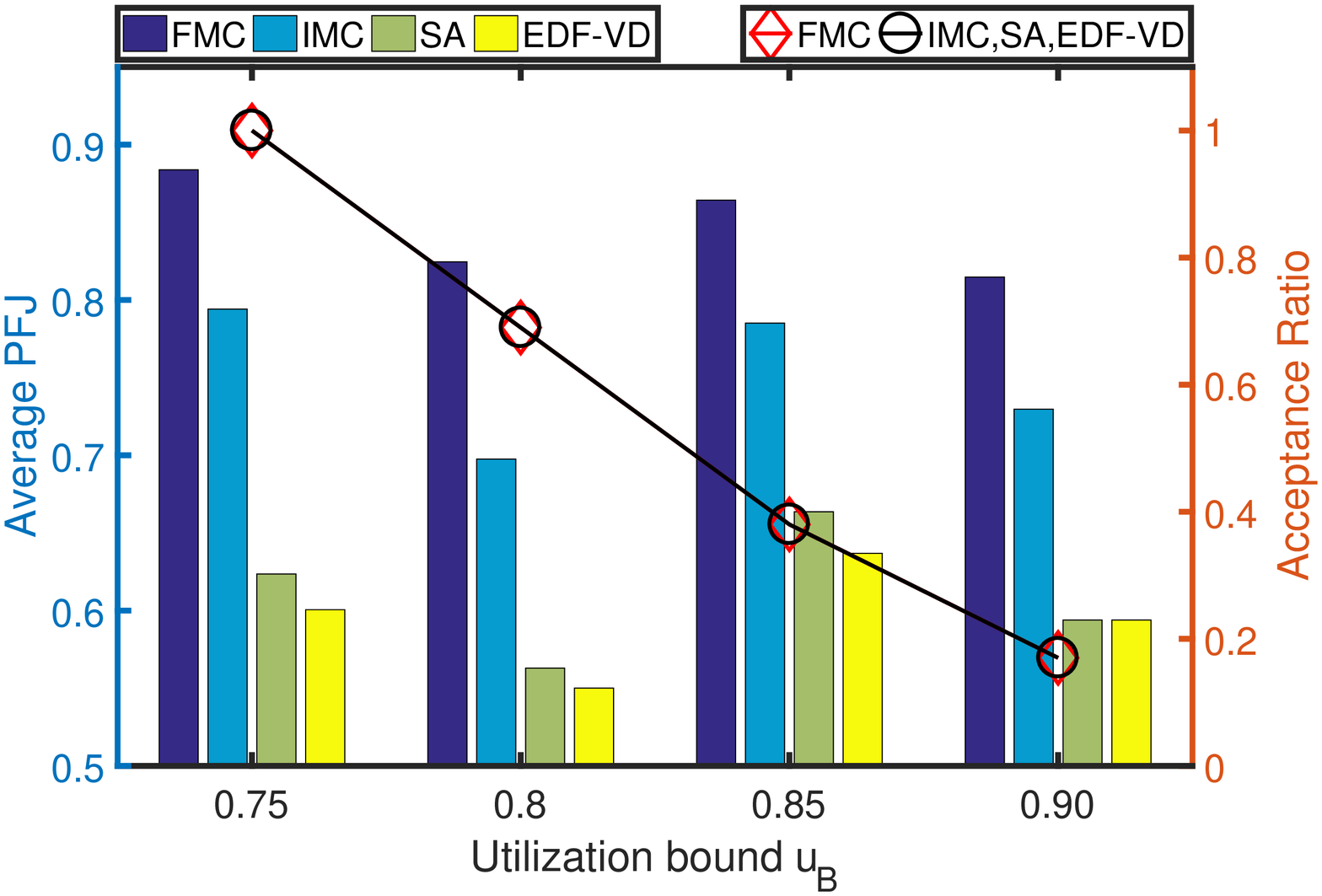}
 \caption{Comparison between FMC and schemes based on the global triggering.}
  \label{fig:res1}
\end{figure}

\subsection{Comparison with schemes based on the global triggering strategy}
First, we demonstrate the effectiveness of FMC compared with the IMC, SA, and classic EDF-VD schemes, which use the global triggering strategy. In these three schemes, \textit{any} overrun will trigger \lo tasks to statically reserve a constant degraded service level. For IMC and FMC, we consider the mandatory utilization $U_{LO}^{man}=0$ for the schedulability test. The schedulability test for the SA scheme~\cite{HGST14a} is a utilization-based test. Therefore, the IMC, SA, and classic EDF-VD schemes have the same schedulability. However, for some schedulable task sets, the SA scheme~\cite{HGST14a} cannot derive a suitable factor $y$ to increase the period of \lo tasks. For this case, we consider $y$ to be infinity, which means that all \lo jobs will be dropped when an overrun occurs.     

For various utilization bounds $u_B\in\{0.75,0.8,0.85,0.9\}$, the average PFJ and system schedulability are compared. The results are shown in \figref{fig:res1}. The left axis shows the PFJ values achieved for \lo tasks, represented by the bar graphs, and the right axis shows the acceptance ratios, represented by the line graphs. From \figref{fig:res1}, we can observe the following trends: (1) FMC consistently outperforms the three other schemes in terms of support for \lo task execution. This is expected because schemes that use the global triggering strategy always consider the worst-case overrun workload, resulting in waste of unnecessary resources. By contrast, FMC can allocate resources based on the true overrun demands. (2) Compared with these three schemes based on the global triggering strategy, FMC achieves almost the same acceptance ratio. This means that FMC can achieve higher on-line \lo performance with negligibly reduced schedulability performance. 
\vskip -2em

\subsection{Comparison with the Pfair- and component-based schemes}
Next, we will experimentally compare our approach to  Pfair- and component-based schemes: PF~\cite{ren} and COM~\cite{gu}. For the component-based scheme COM~\cite{gu}, we use the same experiment setting as~\cite{gu} and consider a two-component system with a \hi component $\mathbb{C_{H}}$ and a \lo component $\mathbb{C_{L}}$. All the \hi tasks are allocated to $\mathbb{C_{H}}$. Each \lo task can be allocated to either $\mathbb{C_{H}}$ or $\mathbb{C_{L}}$\footnote{Since the work presented in~\cite{gu} does not specify the settings for \lo tasks, we specify one probability to determine if a \lo task is allocated to $\mathbb{C_{H}}$. Here, we chosen a relatively low value for probability and set it as 0.25.}. Since the performance of the scheme presented in~\cite{gu} depends on a tolerance limit $TL$, we generate the result of component-based scheme ~\cite{gu} for various values of the tolerance limit $TL=\{0,\floor{0.25|H|},\floor{0.5|H|},\floor{0.75|H|},|H|\}$, where $|H|$ denotes the number of \hi tasks. For the Pfair-based scheme~\cite{ren}, a two-phased scheduling strategy\footnote{the version used for evaluation in~\cite{ren}.} is implemented for comparison. 
\begin{figure}
    \centering
    \subfigure[$u_B=0.75$]{
        \includegraphics[width=0.8\columnwidth,height=0.41\columnwidth]{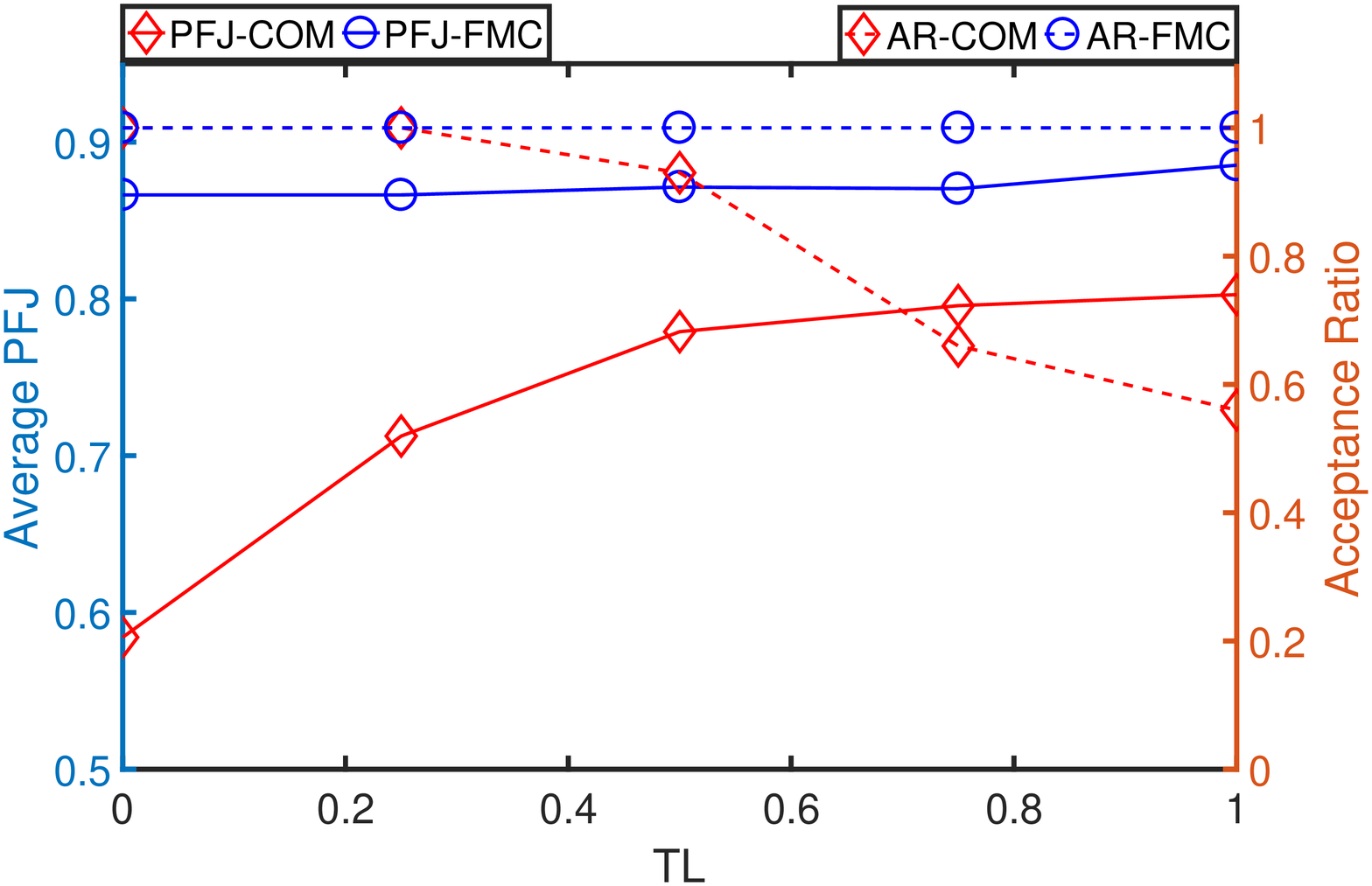}
        \label{fig:08}
    }
    \subfigure[$u_B=0.85$]{
        \includegraphics[width=0.8\columnwidth,height=0.41\columnwidth]{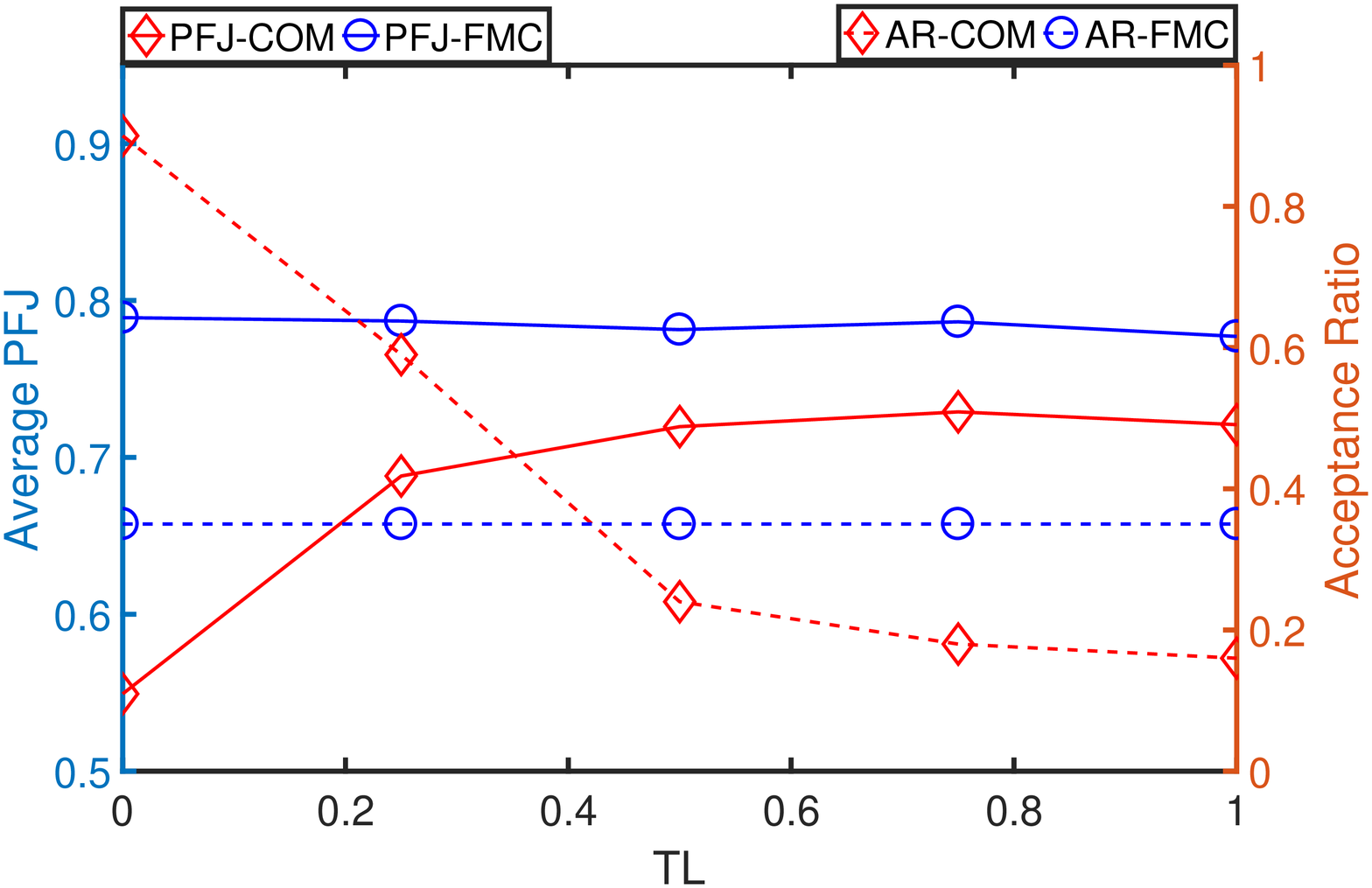}
        \label{fig:09}
    }
		\vskip -0.5em
    \caption{Comparison between component-based scheme and FMC.}
		\label{fig:component}
		\vskip -1em
\end{figure}
 
\noindent\textbf{Comparison with component-based scheme COM~\cite{gu}}: The performance results of COM and FMC are presented in \figref{fig:08} and \figref{fig:09} with different settings on $u_B$. In these figures, x-axis denotes the varying value of $TL$, whereas the left and right y-axis present the average PFJ and acceptance ratio, respectively. As shown in \figref{fig:component}, FMC consistently outperforms COM in terms of support for \lo execution. This performance gain is achieved by the fact that COM adopts pessimistic dropping-off strategies in internal and external mode-switch levels. In COM, all \lo tasks in $\mathbb{C_{H}}$ will be abandoned once \textit{any} overrun occurs and thus results in resource under-utilization. As a comparison, FMC drops off \lo tasks as the demand and therefore can achieve better execution support for \lo tasks. Besides, we can observe that there is a performance trade-off between PFJ and acceptance ratio in COM. The reason for this trend is that a higher $TL$ in COM requires additional resources to support \lo executions but generally implies lower schedulability, and the converse also holds. When we consider the TL configuration on which the same scheduability performance can be achieved by FMC and COM, FMC can support more than 25\% and 15\% \lo tasks to finish the deadline compared to COM under different $u_B$ settings, respectively. 

\begin{figure}
\centering
\includegraphics[width=0.8\columnwidth,height=0.4\columnwidth]{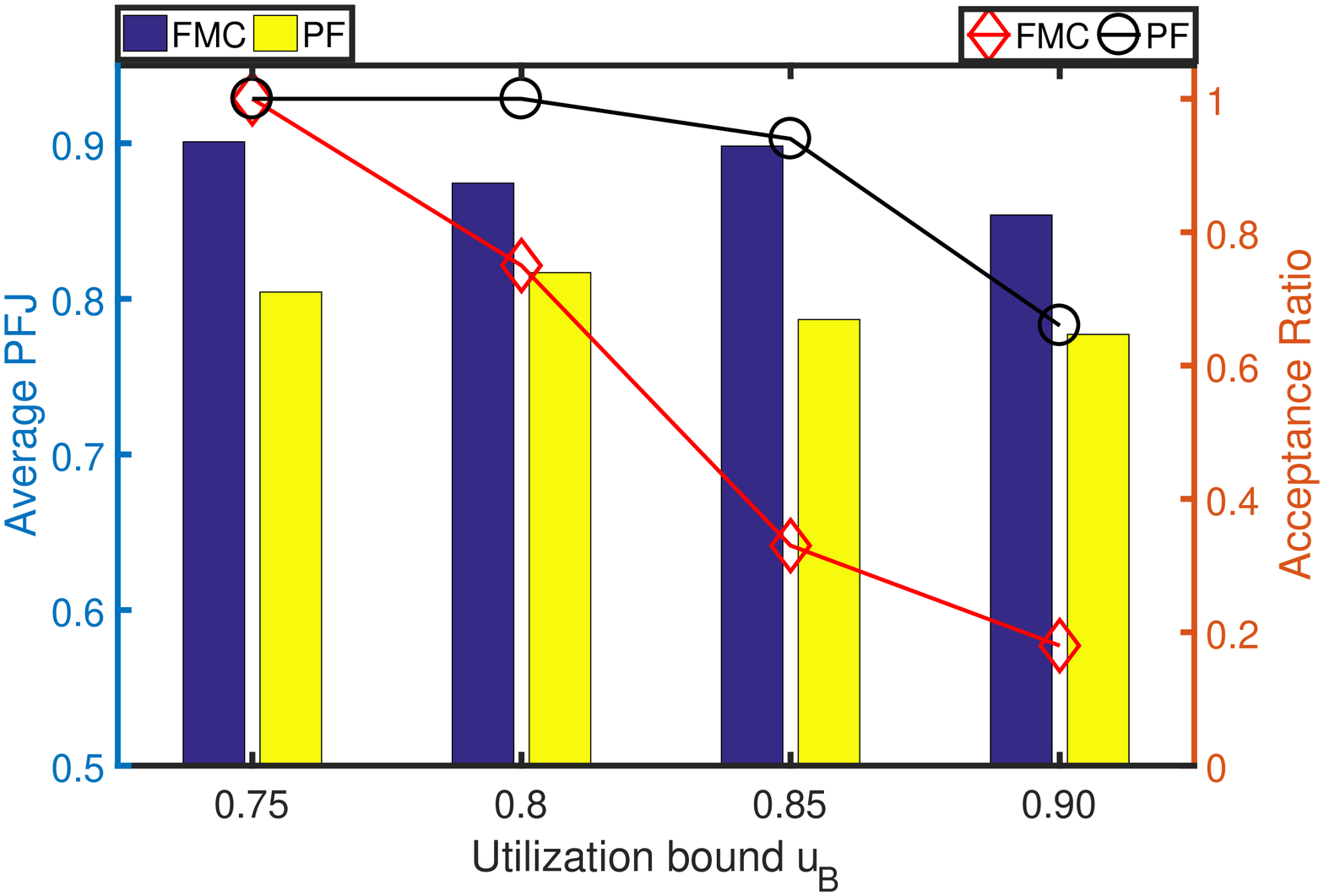}
 \vskip -0.5em
 \caption{Comparison between Pfair-based scheme and FMC.}
  \label{fig:TG}
	\vskip -1em
\end{figure}   

\noindent\textbf{Comparison with Pfair-based scheme PF~\cite{ren}}: \figref{fig:TG} shows the compared results for FMC and PF. Compared with PF, FMC can achieve a better execution support for \lo tasks but with inferior schedulability, as shown in \figref{fig:TG}. The reason for the gain in \lo task execution support is that the Pfair scheduling tends to evenly distribute the quanta of tasks over time, resulting in more unfinished jobs at mode-switching points. 
%Regarding schedulability, we mainly attribute this inferiority to the well-known optimality of Pfair scheduling in terms of schedulability performance~\cite{Pfairs}. However, this optimality of Pfair-based approach highly relies on impractical assumptions, such as ignoring the cost of slicing up tasks and requiring periodical task models. In fact, the schedulability deficit for FMC can be compensated by significantly reduced context-switching overheads compared with Pfair-based approach. A comparison based on this criterion is presented in \figref{fig:switch}.  

\begin{figure}
\centering
\includegraphics[width=0.8\columnwidth,height=0.41\columnwidth]{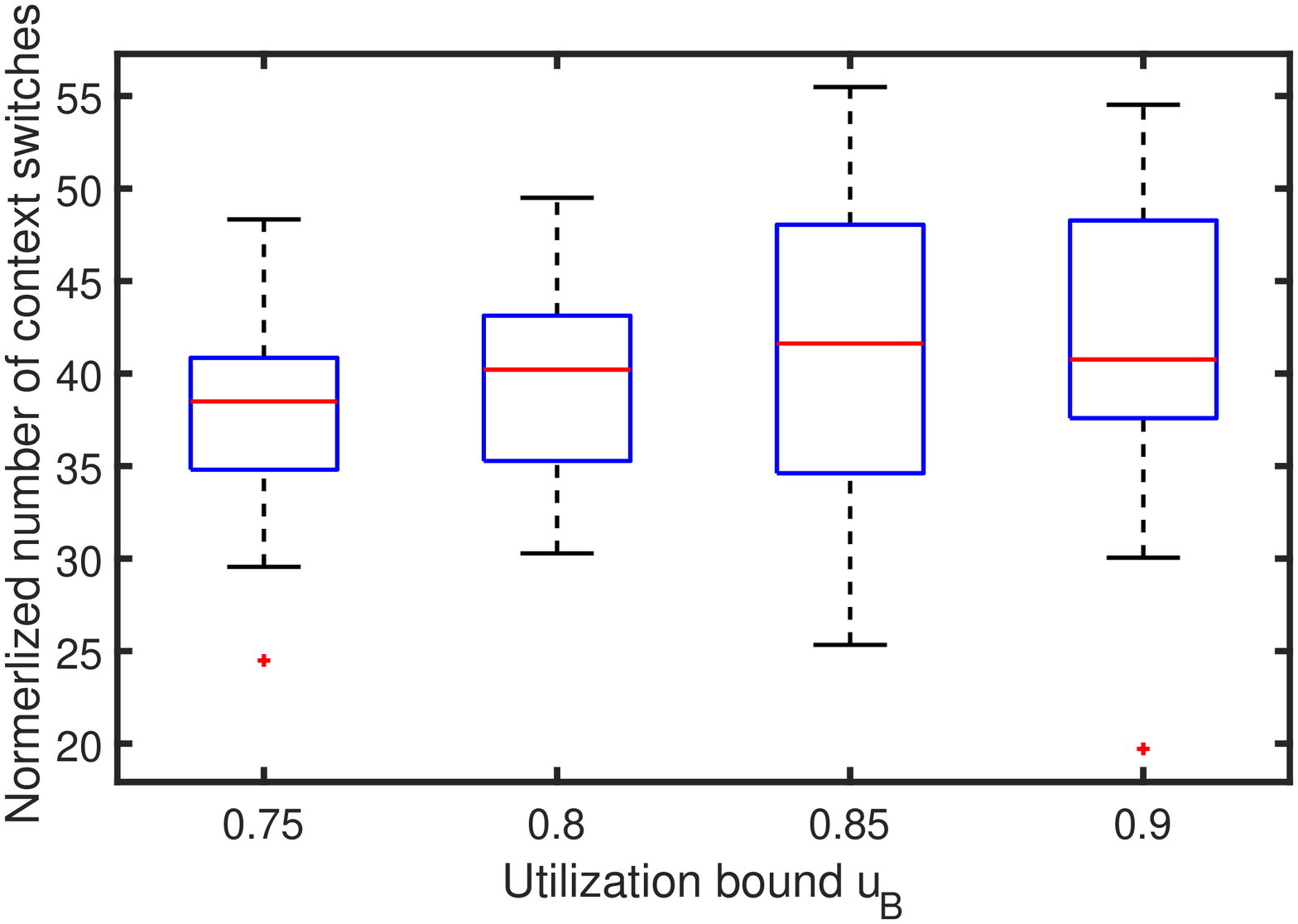}
\vskip -0.5em
 \caption{Context switches for FMC and Pfair-based scheme.}
  \label{fig:switch}
	\vskip -1em
\end{figure}        

{Regarding schedulability inferiority, we mainly attribute this expected inferiority to the theoretical optimality of Pfair scheduling in terms of schedulability performance~\cite{Pfairs}. In fact, this optimality is achieved at the cost of a high scheduling overhead by quantum-length sub-tasks partitioning and the enforcement of proportional progress. In fact, this schedulability deficit of FMC can be compensated by significantly reduced context-switching overheads compared with PF. Here, we present simulation results to show the compared context-switch numbers. 
%We compare the required numbers of context switches for the FMC and Pfair-based scheme. 
\figref{fig:switch} presents the number of context switches for the Pfair-based scheme, which is normalized with respect to the number for FMC. The results confirm
the significant reduction of context switches by FMC. The Pfair-based scheme requires 38.0 to 41.3 times the number of context switches required in FMC for different utilization settings.}
\vskip -2em

\subsection{Graceful service level degradation}
In the case study presented above, we have demonstrated the gradual service level degradation property of FMC, that is, the degradation in the service levels for \lo tasks as the number of mode switch increases. Now, we validate this trend in a generic simulation. The uniform tuning strategy is applied to randomly generated task sets. We use the task generator introduced above to generate 100 task sets in which $u_B$ is randomly selected from $[0.75,0.9]$. A generated task set is accepted for simulation when the following two additional conditions are satisfied: (1) the task set can be scheduled by FMC, and (2) the task set contains 5 \hi tasks. The degraded \lo jobs under various task sets are classified according to the number of mode switches.    
\begin{figure}
\centering
\includegraphics[width=0.8\columnwidth,height=0.41\columnwidth]{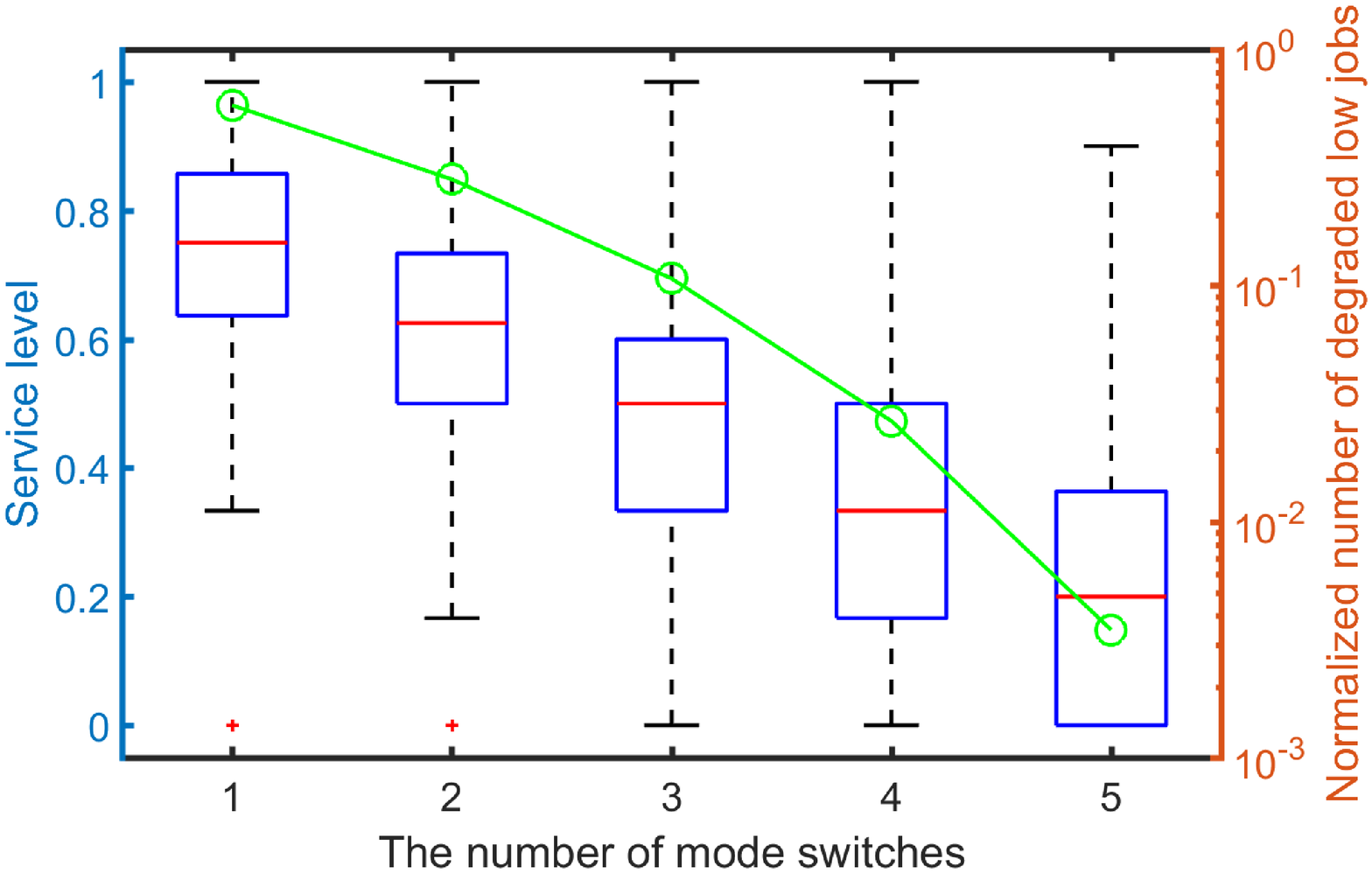}
\vskip -0.5em
 \caption{Service degradation in generic simulation.}
  \label{fig:gs}
	\vskip -1em
\end{figure}

The simulation results are shown in \figref{fig:gs}. The left and right y-axis present the service level and the normalized number of service-degraded \lo jobs, respectively. To reveal the distribution of service levels for service-degraded \lo jobs, the service levels are represented in the form of box-whisker plots. The results shown in \figref{fig:gs} confirm the observations made in \secref{sec:cs}. For almost all \lo task jobs, the graceful degradation property is clearly demonstrated except for a few corner cases. Furthermore, the results in terms of the percentages of service-degraded \lo jobs also confirm that the likelihood of all high-criticality tasks exhibiting the high-criticality behavior is very low. Only 0.35\% of service-degraded \lo jobs are affected by this worst-case overrun scenario. By contrast, 96.9\% of service-degraded \lo jobs are impacted by mode-switch scenarios with mode switches$\leq{}3$. For this vast majority of cases, FMC needs to allocate additional resources to only a subset of the high-criticality tasks based on their demands and therefore can provide better and more graceful service degradation. 
\vskip -2em

%%%%%%%%%%%%%%%%%%%%%%%%%%%%%%%%%%%%%%%%%%%
\section{Conclusion and future work}
\label{sec:conclusion}
Most previous theoretical work on scheduling in mixed-criticality systems has adopted impractical assumptions: once any \hi task overruns, all \lo tasks are suspended and all other \hi tasks are required to exhibit \hi behaviors. 
In this paper, we propose a more flexible MC model
(FMC) with EDF-VD scheduling, in which the above issues are addressed. 
In this model, 
the transitions of all \hi tasks are independent and the service levels of \lo tasks can be adaptively tuned in accordance with the true overruns of the \hi tasks. A utilization-based schedulability test condition is successfully derived for the FMC systems.  Numerical results are presented to illustrate the improved service levels for \lo tasks during run time.

{For the next step, we are interested in implementing the proposed approach on real-time operating system and evaluating its performance. Furthermore, another interesting future work includes investigations on: (1) integrating of FMC and fault tolerance techniques to develop optimal resource allocation strategies for assurances against different types of faults; (2) integrating the slack reclamation schemes into FMC for further performance improvement.}
\vskip -2em
 
{
\footnotesize
\bibliographystyle{abbrv}
\bibliography{mc_1}

\begin{thebibliography}{10}

\bibitem{Baruah2010}
S.~Baruah et~al.
\newblock Towards the design of certifiable mixed-criticality systems.
\newblock In {\em 2010 16th IEEE Real-Time and Embedded Technology and
  Applications Symposium}, 2010.

\bibitem{Baruah2011}
S.~Baruah et~al.
\newblock Mixed-criticality scheduling of sporadic task systems.
\newblock In {\em the 19th European Conference on Algorithms}, 2011.

\bibitem{Baruah2011RTSS}
S.~Baruah et~al.
\newblock Response-time analysis for mixed criticality systems.
\newblock In {\em 2011 IEEE 32nd Real-Time Systems Symposium}, 2011.

\bibitem{Baruah2012}
S.~Baruah et~al.
\newblock The preemptive uniprocessor scheduling of mixed-criticality
  implicit-deadline sporadic task systems.
\newblock In {\em 2012 24th Euromicro Conference on Real-Time Systems}, 2012.

\bibitem{sanjoyACM}
S.~Baruah et~al.
\newblock Preemptive uniprocessor scheduling of mixed-criticality sporadic task
  systems.
\newblock {\em Journal of the ACM}, 62(2), 2015.

\bibitem{RateIMC}
P.~Binns.
\newblock Incremental rate monotonic scheduling for improved control system
  performance.
\newblock In {\em 3rd IEEE Real-Time Technology and Applications Symposium},
  1997.

\bibitem{burns2013towards}
A.~Burns et~al.
\newblock Towards a more practical model for mixed criticality systems.
\newblock In {\em 1st International Workshop on Mixed Criticality Systems},
  pages 1--6, 2013.

\bibitem{burns2015mixed}
A.~Burns et~al.
\newblock Mixed criticality systems-a review.
\newblock {\em University of York, Technical Report}, 2015.

\bibitem{IMCR3}
H.~Chishiro et~al.
\newblock Practical imprecise computation model: Theory and practice.
\newblock In {\em 2014 IEEE 17th International Symposium on
  Object/Component/Service-Oriented Real-Time Distributed Computing}, 2014.

\bibitem{Feng}
W.~Feng et~al.
\newblock An extended imprecise computation model for time-constrained speech
  processing and generation.
\newblock In {\em IEEE Workshop on Real-Time Applications}, 1993.

\bibitem{MINLP}
C.~A. Floudas.
\newblock {\em Deterministic Global Optimization: Theory, Methods and
  Applications}.
\newblock Springer, 2005.

\bibitem{gu}
X.~Gu et~al.
\newblock Resource efficient isolation mechanisms in mixed-criticality
  scheduling.
\newblock In {\em 2015 27th Euromicro Conference on Real-Time Systems}, 2015.

\bibitem{Han2003A}
C.~C. Han et~al.
\newblock A fault-tolerant scheduling algorithm for real-time periodic tasks
  with possible software faults.
\newblock {\em IEEE Transactions on Computers}, 2003.

\bibitem{Pfairs}
P.~Holman et~al.
\newblock Group-based pfair scheduling.
\newblock {\em Real-Time Systems}, 2006.

\bibitem{huang}
P.~Huang et~al.
\newblock Interference constraint graph: A new specification for
  mixed-criticality systems.
\newblock In {\em 2013 IEEE 18th Conference on Emerging Technologies Factory
  Automation}, pages 1--8, 2013.

\bibitem{HGST14a}
P.~Huang et~al.
\newblock Service adaptions for mixed-criticality systems.
\newblock In {\em 2014 19th Asia and South Pacific Design Automation
  Conference}, 2014.

\bibitem{ISO262262}
{ISO 26262:Road vehicles}.
\newblock {\tt http://www.iso.org/iso/}.

\bibitem{LinIMC}
K.-J. Lin et~al.
\newblock Imprecise results: Utilizing partial comptuations in real-time
  systems.
\newblock In {\em Real-Time Systems Symposium}, 1987.

\bibitem{DiLiu}
D.~Liu et~al.
\newblock Edf-vd scheduling of mixed-criticality system with degraded quality
  guarantees.
\newblock In {\em 2016 IEEE 32nd Real-Time Systems Symposium}, 2016.

\bibitem{LinIMC1}
J.~W.~S. Liu et~al.
\newblock Imprecise computations.
\newblock 1994.

\bibitem{Rajkumar1997A}
R.~Rajkumar et~al.
\newblock A resource allocation model for qos management.
\newblock In {\em 1997 IEEE Real-Time Systems Symposium}, 1997.

\bibitem{ren}
J.~Ren et~al.
\newblock Mixed-criticality scheduling on multiprocessors using task grouping.
\newblock In {\em 2015 27th Euromicro Conference on Real-Time Systems}, pages
  25--34, 2015.

\bibitem{Back}
F.~Santy, L.~George, P.~Thierry, and J.~Goossens.
\newblock Relaxing mixed-criticality scheduling strictness for task sets
  scheduled with fp.
\newblock In {\em 2012 24th Euromicro Conference on Real-Time Systems}, 2012.

\bibitem{Su2013}
H.~Su et~al.
\newblock An elastic mixed-criticality task model and its scheduling algorithm.
\newblock In {\em Design, Automation Test in Europe Conference Exhibition},
  2013.

\bibitem{Su2014}
H.~Su et~al.
\newblock Service guarantee exploration for mixed-criticality systems.
\newblock In {\em 2014 IEEE 20th International Conference on Embedded and
  Real-Time Computing Systems and Applications}, 2014.

\bibitem{Vestal2007}
S.~Vestal.
\newblock Preemptive scheduling of multi-criticality systems with varying
  degrees of execution time assurance.
\newblock In {\em 2007 28th IEEE International Real-Time Systems Symposium},
  2007.

\bibitem{ZhuParis}
D.~Zhu et~al.
\newblock Multiple-resource periodic scheduling problem: how much fairness is
  necessary?
\newblock In {\em 24th IEEE Real-Time Systems Symposium}, 2003.

\end{thebibliography}
}

\newpage
\vskip -6em
\begin{appendices}
\section{Derivation protocol}
\label{appendix:II}
This section presents the detailed proofs for the derivation protocol presented in \eqnref{eq:dp1} and \eqnref{eq:dp2}. The rules for deriving \textbf{intermediate} upper bounds for different execution scenarios are specified as follows: 
\begin{rule0}
\label{rule:1}
If no $k$-\textit{carry-over job} of \lo task $\tau_i$ exists at the $k^{\textit{th}}$ mode-switching point $\hat{t}^k$, then the \textbf{intermediate} upper bound $\sups{\eta_i^k(0,t_f)}$ is computed as 
\begin{align}
\sups{\eta_i^k(0,t_f)}=\sups{\eta_i^k(0,d_i^{l})}+(t_f-\hat{t}^k)\cdot{}z_i^k\cdot{}u_i^{LO}
\end{align}
where $d_i^{l}$ denotes the absolute deadline of the last job of $\tau_i$ during $[0,t_f]$.  
\end{rule0}
\begin{proof}
Since no $k$-\textit{carry-over job} of $\tau_i$ exists at $\hat{t}^k$, we know that $d_i^{l}<\hat{t}^k$. Therefore, $\sups{\eta_i^k(0,d_i^{l})}$ is sufficient to bound $\eta_i^k(0,t_f)$. Since $(t_f-\hat{t}^k)\cdot{}z_i^k\cdot{}u_i^{LO}>0$, we know that $\sups{\eta_i^k(0,d_i^{l})}+(t_f-\hat{t}^k)\cdot{}z_i^k\cdot{}u_i^{LO}$ is also an upper bound.                 
\end{proof}
\vskip -2em
\begin{rule0}
\label{rule:2}
If a $k$-\textit{carry-over job} of \lo task $\tau_i$ exists at the $k^{\textit{th}}$ mode-switching point $\hat{t}^k$, then the \textbf{intermediate} upper bound $\sups{\eta_i^k(0,t_f)}$ is computed as
\begin{align}
\sups{\eta_i^k(0,t_f)}=\sups{\eta_i^k(0,d_i^{k})}+(t_f-d_i^{k})\cdot{}z_i^k\cdot{}u_i^{LO}
\end{align}  
\end{rule0}
\begin{proof}
$\sups{\eta_i^k(0,t_f)}$ can be calculated as $\sups{\eta_i^k(0,d_i^{k})}+\sups{\eta_i^k(d_i^{k},t_f)}$. Since $(t_f-d_i^{k})\cdot{}z_i^k\cdot{}u_i^{LO}$ is sufficient to bound $\eta_i^k(d_i^{k},t_f)$ according to \propref{prop:x}, we know that $\sups{\eta_i^k(0,d_i^{k})}+(t_f-d_i^{k})\cdot{}z_i^k\cdot{}u_i^{LO}$ can be used as an upper bound.               
\end{proof}
\vskip -2em     
\begin{rule0}
\label{rule:3}
For a $k$-\textit{carry-over job} of \lo task $\tau_i$, if $\eta_i^k(a_i^k,\hat{t}^{k})\neq{0}$, then the \textbf{intermediate} upper bound $\sups{\eta_i^k(a_i^k,d_i^{k})}$ will be computed as follows: 
\begin{align}
\sups{\eta_i^k(a_i^k,d_i^{k})}={(d_i^k-a_i^k)\cdot{z_i^{k-j}}\cdot{u_i^{LO}}}
\end{align}
where $z_i^{k-j}$ is the service level after the last mode switch occurring before $a_i^k$.  
\end{rule0}

\begin{proof}
According to \figref{fig:prop}, $j$ mode switches may occur during the interval $(a_i^k, \hat{t}^k)$. Recall the assumption $z_i^k\le{z_i^{k-1}}$ ($\forall k$) made in \secref{sec:model}; based on this assumption, we have $z_i^{k-j}\ge{}z_i^{k-(j-1)}\ge{}\cdots\ge{}z_i^{k}$. When the system switches modes at $\hat{t}^{k_0}$ ($k-(j-1)\le{k_0}\le{k}$), the execution budget will be reduced from $z_i^{k_0-1}\cdot{c_i^{LO}}$ to $z_i^{k_0}\cdot{c_i^{LO}}$. Therefore, the maximum cumulative execution is achieved when the $k$-\textit{carry-over job} has completed its $z_i^{k-j}\cdot C_i^{LO}$ execution before time instant $\hat{t}^{k-(j-1)}$, which is the time of the first mode switch occurring after $a_i^k$. Therefore, in this case, the $k$-\textit{carry-over job} can be bounded by $(d_i^k-a_i^k)\cdot{z_i^{k-j}}\cdot{u_i^{LO}}$. 
\end{proof}
\vskip -4em
\begin{rule0}
\label{rule:4}
For a $k$-\textit{carry-over job} of \lo task $\tau_i$, 
if $\eta_i^k(a_i^k,\hat{t}^{k})={0}$, the \textbf{intermediate} upper bound $\sups{\eta_i^k(a_i^k,d_i^{k})}$ will be computed as follows: 
\begin{align}
\sups{\eta_i^k(a_i^k,d_i^{k})}={(d_i^k-a_i^k)\cdot{z_i^{k}}\cdot{u_i^{LO}}}
\end{align}
\end{rule0}
\begin{proof}
Since the $k$-\textit{carry-over job} has not been executed before $\hat{t}^{k}$ ($\eta_i^k(a_i^k,\hat{t}^{k})={0}$), it must exhaust its execution budget $z_i^k\cdot{}C_i^{LO}$ after the mode-switching time instant $\hat{t}^k$. Therefore, the maximum cumulative execution at the current time can be found to be $(d_i^k-a_i^k)\cdot{z_i^{k}}\cdot{u_i^{LO}}$.
\vskip -8em
\end{proof}
\vskip -8em
\section{The derivation of $N^{k}_{\gamma}$ in \eqnref{eq:main:0}}
\label{appendix:I}
The following is the detailed derivation of the upper bound on the total cumulative execution time $N^{k}_{\gamma}$.
\begingroup
\allowdisplaybreaks
\begin{scriptsize}
\begin{align}
&N^{k}_{\gamma}=N^{k}_{\gamma_{LO}}+N^{k}_{\gamma_{HI}^{HI}(\hat{t}^{k})}+N^{k}_{\gamma_{HI}^{LO}(\hat{t}^{k})} \nonumber \\
\le{}&\underbrace{\sum_{\tau_i\in\gamma_{LO}}\bigg(t_fu_i^{LO}+\sum_{j=1}^{k}{((t_f-a_{\hat{t}^{j}})(1-x)(z_i^{j}-z_i^{j-1})u_i^{LO})}\bigg)}_{N^{k}_{\gamma_{LO}}}  \nonumber \\
&+\underbrace{\sum_{j=1}^{k}{(a_{\hat{t}^j}\cdot{}u_{\hat{t}^j}^{LO}+(t_f-a_{\hat{t}^j})\cdot{}u_{\hat{t}^j}^{HI})}}_{N^{k}_{\gamma_{HI}^{HI}(\hat{t}^{k})}} 
+\underbrace{\sum_{\tau_i\in\gamma_{HI}^{LO}}{\frac{t_f}{x}\cdot{}u_{i}^{LO}}}_{N^{k}_{\gamma_{HI}^{LO}(\hat{t}^{k})}} \nonumber \\
=&\underline{t_fu_{LO}^{LO}}+\sum_{j=1}^{k}{(t_f-a_{\hat{t}^{j}})(1-x)(u_{LO}^{j}-u_{LO}^{j-1})}\nonumber \\
&+\sum_{j=1}^{k}{(a_{\hat{t}^j}\cdot{}u_{\hat{t}^j}^{LO}+(t_f-a_{\hat{t}^j})\cdot{}u_{\hat{t}^j}^{HI})} 
+\sum_{\tau_i\in\gamma_{HI}^{LO}}{\frac{t_f}{x}\cdot{}u_{i}^{LO}} \nonumber \\
=&\ (Consider\ a_0=\frac{\sum_{j=1}^{k}{a_{\hat{t}^{j}}u_{\hat{t}^j}^{LO}}}{\sum_{j=1}^{k}{u_{\hat{t}^j}^{LO}}})   \nonumber \\
&\underline{(t_f-a_0)u_{LO}^{LO}+a_0u_{LO}^{LO}}+\sum_{j=1}^{k}{(t_f-a_{\hat{t}^{j}})(1-x)(u_{LO}^{j}-u_{LO}^{j-1})}\nonumber \\
&+\sum_{j=1}^{k}{(a_{\hat{t}^j}\cdot{}u_{\hat{t}^j}^{LO}+(t_f-a_{\hat{t}^j})\cdot{}u_{\hat{t}^j}^{HI})} 
+\sum_{\tau_i\in\gamma_{HI}^{LO}}{\frac{t_f}{x}\cdot{}u_{i}^{LO}} \nonumber \\
\le&\ (Since\ u_{LO}^{LO}+\frac{u_{HI}^{LO}}{x}\le{1}\ and\ x<1) \nonumber \\
&(t_f-a_0)u_{LO}^{LO}+\underline{a_0(1-\frac{u_{HI}^{LO}}{x})}+\sum_{j=1}^{k}{(t_f-a_{\hat{t}^{j}})(1-x)(u_{LO}^{j}-u_{LO}^{j-1})}\nonumber \\
&+\sum_{j=1}^{k}{(\underline{\frac{a_{\hat{t}^j}}{x}u_{\hat{t}^j}^{LO}}+(t_f-a_{\hat{t}^j})u_{\hat{t}^j}^{HI})} 
+\sum_{\tau_i\in\gamma_{HI}^{LO}}{\frac{t_f}{x}u_{i}^{LO}} \nonumber \\
=&\ (Since\ u_{HI}^{LO}=\sum_{j=1}^{k}{u_{\hat{t}^j}^{LO}}+\sum_{\tau_i\in\gamma_{HI}^{LO}}{u_{i}^{LO}}) \nonumber \\
&t_f+(t_f-a_0)(u_{LO}^{LO}-1)+\sum_{j=1}^{k}{(t_f-a_{\hat{t}^{j}})(1-x)(u_{LO}^{j}-u_{LO}^{j-1})}\nonumber \\
+&\underline{\sum_{j=1}^{k}{\frac{a_{\hat{t}^j}u_{\hat{t}^j}^{LO}-a_0u_{\hat{t}^j}^{LO}}{x}}}+\sum_{j=1}^{k}{(t_f-a_{\hat{t}^j})u_{\hat{t}^j}^{HI}} 
+\underline{\sum_{\tau_i\in\gamma_{HI}^{LO}}{\frac{t_f-a_0}{x}u_{i}^{LO}}} \nonumber \\
=&\ (Since\ \sum_{j=1}^{k}{(a_{\hat{t}^j}u_{\hat{t}^j}^{LO}-a_0u_{\hat{t}^j}^{LO})}=0)\nonumber \\
&t_f+(t_f-a_0)(u_{LO}^{LO}-1)+\sum_{j=1}^{k}{(t_f-a_{\hat{t}^{j}})(1-x)(u_{LO}^{j}-u_{LO}^{j-1})}\nonumber \\
&+\sum_{j=1}^{k}{(t_f-a_{\hat{t}^j})u_{\hat{t}^j}^{HI}} +\sum_{\tau_i\in\gamma_{HI}^{LO}}{\frac{t_f-a_0}{x}u_{i}^{LO}} \nonumber \\
=&\ (Since\ t_f-a_0=\frac{\sum_{j=1}^{k}(t_f-a_{\hat{t}^j})u_{\hat{t}^j}^{LO}}{\sum_{j_0=1}^{k}u_{\hat{t}^{j_0}}^{LO}})\nonumber \\
&t_f+\sum_{j=1}^{k}(t_f-a_{\hat{t}^j})\bigg(\frac{u_{\hat{t}^j}^{LO}}{\sum_{j_0=1}^{k}u_{\hat{t}^{j_0}}^{LO}}(u_{LO}^{LO}-1)+(1-x)(u_{LO}^{j}-u_{LO}^{j-1}) \nonumber \\
&+u_{\hat{t}^j}^{HI}+\frac{u_{\hat{t}^j}^{LO}}{\sum_{j_0=1}^{k}u_{\hat{t}^{j_0}}^{LO}}\sum_{\tau_i\in\gamma_{HI}^{LO}}\frac{u_{i}^{LO}}{x}\bigg) \nonumber\\
=&\ (Since\ u_{HI}^{LO}=\sum_{j=1}^{k}{u_{\hat{t}^j}^{LO}}+\sum_{\tau_i\in\gamma_{HI}^{LO}}{u_{i}^{LO}})\nonumber \\
&t_f+\sum_{j=1}^{k}(t_f-a_{\hat{t}^j})\bigg(u_{\hat{t}^j}^{HI}-\frac{u_{\hat{t}^{j}}^{LO}}{x}+(1-x)(u_{LO}^{j}-u_{LO}^{j-1}) \nonumber \\
&+\frac{u_{\hat{t}^j}^{LO}}{\sum_{j_0=1}^{k}u_{\hat{t}^{j_0}}^{LO}}(u_{LO}^{LO}-1+\frac{u_{HI}^{LO}}{x})\bigg) \nonumber\\
\le&\ (Since\ u_{LO}^{LO}-1+\frac{u_{HI}^{LO}}{x}\le{0} \ and \ \sum_{j_0=1}^{k}u_{\hat{t}^{j_0}}^{LO}\le{}u_{HI}^{LO})\nonumber \\
&t_f+\sum_{j=1}^{k}(t_f-a_{\hat{t}^j})\bigg(u_{\hat{t}^j}^{HI}-\frac{u_{\hat{t}^{j}}^{LO}}{x}+(1-x)(u_{LO}^{j}-u_{LO}^{j-1}) \nonumber \\
&+\frac{u_{\hat{t}^j}^{LO}}{u_{HI}^{LO}}(u_{LO}^{LO}-1+\frac{u_{HI}^{LO}}{x})\bigg) \nonumber \\
=&t_f+\sum_{j=1}^{k}(t_f-a_{\hat{t}^j})\bigg(u_{\hat{t}^j}^{HI}+(1-x)(u_{LO}^{j}-u_{LO}^{j-1})+\frac{u_{\hat{t}^j}^{LO}}{u_{HI}^{LO}}(u_{LO}^{LO}-1)\bigg) \nonumber
\end{align}
\vskip -2em
\end{scriptsize}
\endgroup
\vskip -2em
\section{Simulation Framework}
\label{appendix:III}
{
In this section, we will present the simulation framework which faithfully emulate real time execution behaviors of mixed criticality tasks. This simulation framework enables us to evaluate various advanced mixed criticality scheduling algorithms and to obtain performance statistics for algorithms evaluation. 
%As shown in \figref{}, we give an overview of the simulation framework. 
In the simulation, the workload trace of jobs are firstly generated. All the compared approaches are tested by the same workload trace in the simulation, such that the fairness can be guaranteed. During the simulation, the statistics entity will collect data from the simulation to calculate different metrics for the compared algorithms. In our simulation, as long as \lo job do not finish its $C_i^{LO}$, this \lo job will be considered as the job without completion. The main aspects of the simulation process will be described as follows.

\noindent\textbf{Workload Trace Generation}: For the sake of fairness, we generate the workload trace of jobs according to task specification and use this workload trace to uniformly test the compared approaches. In this process, the arrival time, execution time, and overruns of jobs are generated before the simulation. To validate the runtime scheduability in the worst-case scenario, all jobs are released with minimum job-arrival interval and are executed with the worst-case execution time for stress testing. The overruns of each \hi jobs are randomly generated according to the overrun probability, which is computed according to execution distribution presented in \cite{Back}. 
%For the sake of fairness, we generate the workload trace of jobs according to task specification and use this .  }

\noindent\textbf{Run-time Simulation}: The system tick is the time unit that causes the scheduler to run and process events. In our simulation, we emulate real time execution behaviors of mixed criticality tasks based on system tick. The possible events are processed according to the following rules:

\begin{itemize}
\item \textbf{Newly-arrived job}: If the system is in \lo mode, the newly-arrived job will be inserted into the queue $Q$ with modified deadline and unmodified execution budgets. Otherwise, \hi jobs with be inserted with unmodified deadline and execution budgets $C_i^H$, while \lo jobs will be inserted with degraded execution budgets $z_i^k\cdot{}C_i^L$. 
\item \textbf{Job completion}: If the job reaches its execution budget, delete the job from the queue $Q$. Record the performance data of the jobs.   
\item \textbf{Preemption and consume time}: Find the job with minimum deadline in $Q$ to execute and consume the current time slot. If the current job is different with the one executed in last time unit, the preemption is identified.    
\item \textbf{Overrun and mode switches}: If \hi tasks overrun its execution budgets $C_i^{LO}$, transit system to \hi mode and update the service level $z_i^k$ for each \lo task.
\item \textbf{Switch back}: If $Q$ is empty, the idle interval is detected. The system will be transit back to \lo mode by resetting $z_i^k=1$ for each \lo task.    
\end{itemize}
Similarly, this simulation framework can be easily extended for other compared approaches by replacing the degradation strategy in the simulation.
\vskip -1em
\section{The screen-shot of submission}
The submission history is shown in \figref{fig:sh}.
\begin{figure}
\centering
\includegraphics[width=0.9\columnwidth,height=0.5\columnwidth]{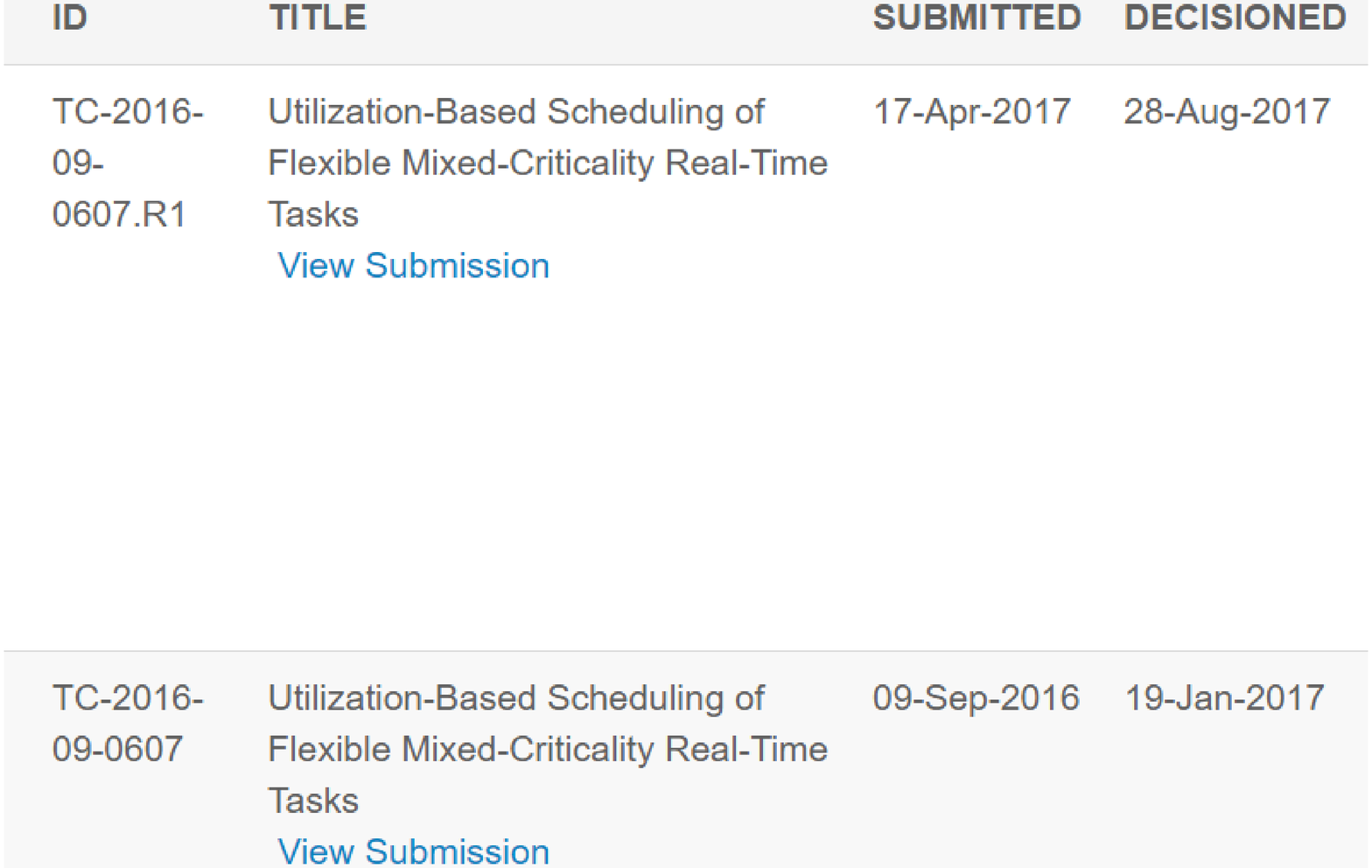}
\vskip -0.5em
 \caption{Submission Information on TC.}
  \label{fig:sh}
	\vskip -1em
\end{figure}

}
\end{appendices}

\end{document}